\documentclass[lettersize,journal, 10pt]{txns_template/IEEEtran}
\usepackage{algorithm}
\usepackage{array}
\usepackage[caption=false,font=normalsize,labelfont=sf,textfont=sf]{subfig}
\usepackage{textcomp}
\usepackage{stfloats}
\usepackage{url}
\usepackage{verbatim}
\usepackage{graphicx}
\usepackage{amsmath,amssymb,amsfonts, amsthm}
\usepackage{algorithmic}
\usepackage{url}
\usepackage{textcomp}
\usepackage{xcolor}
\usepackage{comment}
\usepackage{hyperref}
\hypersetup{
  colorlinks=true,
  urlcolor=blue
}

\newtheorem{proposition}{Proposition}

\newtheorem{lemma}{Lemma}

\newtheorem{assumption}{Assumption}

\DeclareMathOperator*{\argmax}{arg\,max}
\DeclareMathOperator*{\argmin}{arg\,min}
\usepackage[a-1b]{pdfx}  

\usepackage{biblatex} 
\usepackage{hyperref}

\begin{document}
\title{Pricing for Routing and Flow-Control \\ in Payment Channel Networks}
\author{Suryanarayana Sankagiri and Bruce Hajek, \IEEEmembership{Fellow, IEEE}
\thanks{This work was supported in part by the NSF Grant CCF 19-00636.}
\thanks{Suryanarayana Sankagiri is currently a postdoc in the School of Computer and Communication Sciences at EPFL in Lausanne, Switzerland. (email: suryanarayana.sankagiri@epfl.ch).}
\thanks{Bruce Hajek is the Hoeft chair professor in 
the Electrical and Computer Engineering Department and the Coordinated Science Laboratory, University of Illinois, Urbana-Champaign, IL 
61801 USA (email: b-hajek@illinois.edu).}}

\maketitle

\begin{abstract}
A payment channel network is a blockchain-based overlay mechanism that allows parties to transact more efficiently than directly using the blockchain. These networks are composed of payment channels that carry transactions between pairs of users. Due to its design, a payment channel cannot sustain a net flow of money in either direction indefinitely. Therefore, a payment channel network cannot serve transaction requests arbitrarily over a long period of time. We introduce \emph{DEBT control}, a joint routing and flow-control protocol that guides a payment channel network towards an optimal operating state for any steady-state demand. In this protocol, each channel sets a price for routing transactions through it. Transacting users make flow-control and routing decisions by responding to these prices. A channel updates its price based on the net flow of money through it. We develop the protocol by formulating a network utility maximization problem and solving its dual through gradient descent. We provide convergence guarantees for the protocol and also illustrate its behavior  through simulations.
\end{abstract}


\section{Introduction}\label{sec:introduction}

Blockchains, pioneered by Bitcoin, are systems that support digital transactions in a completely decentralized manner. However, most blockchains have poor transaction throughput, a fundamental limitation that stems from their decentralized design (e.g., Bitcoin processes around ten transactions per second \cite{croman2016scaling}). This low throughput results in exorbitant transaction fees and hinders widespread adoption. To be a viable option in practice, blockchain throughput must scale significantly. \textit{Layer-two blockchain mechanisms} are tools that allow many transactions to take place outside of the underlying blockchain system, thereby increasing the system's throughput. See \cite{gudgeon2020sok} for a comprehensive survey of these methods. A payment channel network (PCN) is one such layer-two mechanism that is used in practice. Recent years have seen considerable research interest on PCNs, with a focus on improving their security as well as their efficiency. This paper focuses on their long-term transaction processing efficiency.

As the name suggests, a PCN is a network composed of multiple \textit{payment channels}. A payment channel is a special account that two parties jointly create by depositing some funds. Once a channel is established, the parties transact by exchanging digitally signed messages between themselves, without recording these transactions on the blockchain. A PCN consists of many such payment channels operating together, allowing users who do not share a channel to direct their transactions through intermediaries. Thus, a PCN facilitates a much larger volume of transactions for the same amount of escrowed funds than what would be possible through standalone payment channels.

Although a payment channel can support indefinitely many transactions, it imposes some constraints.  A steady flow of transactions through a payment channel in one direction depletes the funds available at the source. In the long term, a channel can only sustain a \textit{balanced flow} of transactions; \textit{i.e.}, the flows of money in each direction are equal. These balance constraints apply to each channel in a PCN, potentially limiting the PCN's ability to meet the transaction demands of its users. The central problem we address in this work is how the network can serve the transaction demands it receives to the best extent possible, while maintaining \textit{detailed balanced flows} (balanced flows on each channel). 

In this work, we propose the \textit{DEBT (DEtailed Balance Transaction) control protocol}, a joint flow-control and routing protocol for PCNs. The protocol uses channel prices as a decentralized control mechanism for flows. The prices penalize flows that increase the degree of imbalance through a channel and incentivize flows in the opposite direction. Transacting users calculate the path price (the sum of channel prices along a path) for different paths between them. By selecting the path with the minimum path price to execute the transaction, the users perform \textit{routing}. In addition, they adjust the transaction amount as a function of the minimum price, thereby performing \textit{flow-control}. Channels update their prices over time based on the net flow of transactions through them. 

Our main contribution is to prove that given an arbitrary demand and an arbitrary PCN topology, the prices and the flows under this protocol converge to an optimal value in the following sense. The flows maximize the total utility of all users (plus a regularizer term), subject to the detailed balance condition. Loosely, this translates to the PCN serving as many transactions as possible over the long run. 

We begin this paper by introducing the basic terminology for payment channels in Section \ref{sec:payment_channels} and reviewing the literature on payment channel networks in Section \ref{sec:pcn_review}. With this background, we describe the main contributions of this paper in Section \ref{sec:contribution}. In Section \ref{sec:PCN_model}, we present a discrete-time model of a payment channel network, specifying the order of events in a time slot, the nature of transaction requests, and the feasibility constraints imposed by the network. Section \ref{sec:debt_control_protocol} is devoted to the design of the DEBT control protocol. The protocol's convergence follows under some additional assumptions; this is presented in Section \ref{sec:convergence}. Simulation results of the protocol on some simple PCNs are shown in Section \ref{sec:pcn_simulations}. These examples illustrate how the protocol performs routing and flow-control by reacting to prices. We conclude this paper in Section \ref{sec:discussion} with a discussion of the implications of the assumptions we have made in our analysis and the potential practical impact of our work.

\section{Background and Related Work}\label{sec:background}

\subsection{Payment Channels: Basic Principles}\label{sec:payment_channels}

The idea of a payment channel was born at least as early as 2013 and has since undergone considerable evolution. We refer the interested reader to Section 3 of the survey \cite{gudgeon2020sok} and Chapter 5 of the thesis \cite{tikhomirov2020security} for a historical survey of payment channels. Here, we describe the construction and operation of a payment channel at a level of generality that encompasses several designs. 

\paragraph{Construction and Basic Operation} A payment channel is created by means of a special \textit{funding transaction}, with two nodes depositing some money into a new two-node account (the channel). This funding transaction is recorded on the blockchain. At the time of creation, both nodes also create and privately hold a \textit{commitment transaction} which disburses the money held in the channel back to their personal accounts.
Once a channel is created between two nodes, they can transact by exchanging messages between themselves. Each message is simply a new commitment transaction, \textit{i.e.}, an agreement between the two nodes to split the escrowed fund in a certain portion between the two nodes. None of these transactions are broadcast to the blockchain network. A channel is closed by recording a commitment transaction on the blockchain that returns the appropriate balances to the two nodes' individual accounts. A payment channel also includes some inbuilt safety mechanisms by which an honest node can withdraw its funds even if the other node goes missing or acts maliciously.

\paragraph{Channel Capacity and Balances} The total amount of money escrowed in a channel is called the \textit{capacity} of the channel. The amount of money each node owns within the channel are the \textit{balances} of the two nodes. While the capacity remains constant over time, the balances change with every transaction. The sum of the balances always equals the capacity. At any given point in time, a channel's balances impose a bound on the maximum value of a single transaction that can be made in either direction. To elaborate, consider a channel between two nodes $A$ and $B$, and let their balances be $x_A$ and $x_B$ respectively. Then $A$ cannot pay $B$ more than $x_A$ and $B$ cannot pay $A$ more than $x_B$ amount of money in a single transaction. Usually, the transaction values through a payment channel are much smaller than the channel's capacity. Therefore, if a channel is \textit{balanced}, \textit{i.e.}, the balances are each close to half the capacity, it will be able to execute any transaction request that might arise from either end.

\paragraph{Flows} Consider a scenario where node $A$ pays node $B$ ten dollars each day over the channel $A-B$ (\textit{e.g.}, a customer buying coffee from a cafe every day). This long-term transaction pattern can be represented as a \textit{flow} of money across the channel. More generally, a flow through a channel represents a long-term transfer of money from one node to another. We can use the notion of a flow to represent a series of transactions whose values fluctuate randomly or periodically around some average value.

The notion of flows sheds some light on the choice of the term \textit{channel capacity}, defined above. In the context of communication networks, this term denotes the maximum rate of communication that a channel can support, \textit{i.e.}, the maximum number of bits (or information) that the channel can carry per unit time. For a payment channel, there is a similar interpretation; its capacity, \textit{i.e.,} the total escrowed fund, is proportional to the maximum possible flow rate through the channel (sum of the flows in each of the two directions). Indeed, the flow through a channel is maximized if each node sends its complete balance to the counter-node immediately after the previous transaction is completed and the resulting flow rate is equal to the total escrowed fund per time slot. 

\paragraph{The Balance Condition for Flows} In addition to a bound on the total flow rate, a payment channel also imposes a \textit{balanced flow constraint}, which can be stated as follows. Without persistent on-chain rebalancing (described below), a payment channel cannot sustain a nonzero, constant net flow indefinitely. Suppose, for instance, that on average, $A$ pays $B$ ten dollars per day and $B$ pays $A$ five dollars per day. Over time, $A$'s balance in the channel decreases continuously. Whatever be the channel's capacity, eventually $A$'s balance gets completely exhausted. If $A$'s transactions to $B$ are \textit{elastic} (\textit{i.e.}, $A$ is willing to tolerate a smaller flow rate than its original demand), $A$ could wait for $B$ to make some payments to $A$. Effectively, the flow rate from $A$ to $B$ drops by half to five dollars per day. Thus, the channel settles to a state where it has zero net flow through it. This example illustrates how a channel naturally enforces the balance constraint on flows through it.

\paragraph{On-Chain Rebalancing} If flows are \textit{inelastic}, \textit{i.e.}, the entire demand must be served through the channel, then the channel can reset its balances via an on-chain transaction. One option is that $A$ pays $B$ some amount via an on-chain transaction, with $B$ paying $A$ back the same amount on the channel. If this amount is half the channel capacity, the entire operation would rebalance the channel. Alternatively, the nodes may open a new, perfectly balanced channel between themselves. In summary, with persistent on-chain rebalancing, a channel can support a nonzero net flow. However, on-chain rebalancing is an expensive operation and should be avoided to the extent possible. Moreover, supporting a nonzero net flow through a channel with persistent on-chain rebalancing is equivalent to serving balanced flows on the payment channel and serving the rest of the demand on the blockchain.

\subsection{Routing in Payment Channel Networks: A Review}\label{sec:pcn_review}
Payment channel networks are a set of payment channels operating in cooperation. The most popular such network is the Lightning Network for Bitcoin, proposed by \cite{poon2016bitcoin}. At the time of writing, the network has more than $10, 000$ nodes, $40,000$ channels, and $5,000$ Bitcoin escrowed in the channels (see \url{1ml.com} for live statistics). The Lightning network has played a central role in the growth of small-scale transactions via Bitcoin \cite{river2023lightning}. Alongside the surge in popularity of the Lightning network, there has been a growing academic literature on payment channel networks. Various  aspects of payment channel networks have been investigated, such as their security, privacy, and efficiency \cite{gudgeon2020sok, papadis2020blockchain}. 

In a payment channel network, the channel capacities and the topology of the network are known to all nodes, but the real-time balances are not. While this lack preserves privacy, it poses a significant challenge in even the simplest task of finding a \textit{feasible path} for a transaction \cite{tang2020privacy}. (A path connecting two transacting nodes is said to be feasible if each channel along the path has sufficient balance to execute the transaction.) A naive approach for the path discovery problem, used by the Lightning Network, is the following. One attempts to transact along the shortest path in the network (which can be found through the known topology); if the attempt fails, the faulty edge is removed from the view and again a shortest path is found. With deeper research into routing strategies for PCNs, it was soon established that the naive routing strategy used by the Lightning network is suboptimal.

There are many routing protocols that have been proposed in the literature. We summarize some of the important findings of this literature that are related to our work.
\begin{itemize}
    \item Smaller transactions are more likely to succeed than larger ones, simply because their feasibility requirements are lower. Therefore, it is worth splitting a large transaction into multiple smaller components and sending these along different paths \cite{stasi2018routing, sivaraman2020high}.
    \item A PCN's throughput can suffer from channel congestion. For example, if all nodes choose to transact along the shortest path between them, channels with high centrality (\textit{i.e.,} with many shortest paths through them) may get congested; the total demand through them could exceed their capacity (see discussion in Section \ref{sec:payment_channels}). These channels will be forced to drop transactions because they do not have the balance to serve them \cite{sivaraman2020high, liu2023balanced}.
    \item Another factor affecting a PCN's throughput is that channels can get imbalanced. As mentioned in Section \ref{sec:payment_channels}, a channel through which there is a net flow of money from one end to another eventually gets imbalanced and is unable to support further flows in that direction. An imbalanced channel could potentially stifle out many end-to-end flows, lowering the PCN's throughput \cite{lin2020funds, van2021merchant}.
    \item Smart routing can improve a PCN's throughput by mitigating both congestion and imbalance. A routing protocol can discover paths that are slightly longer than the shortest path where channels are underutilized (or have higher capacity), thereby reducing congestion \cite{liu2023balanced}. Similarly, a routing protocol can also direct transactions along paths in a manner so as to rebalance channels \cite{stasi2018routing, ren2018optimal, lin2020funds, van2021merchant}.
    \item Using channel prices/fees as a means of routing is an economically sound and practically viable approach. Each channel can charge a price that dynamically changes with the extent of imbalance or congestion it faces. In particular, a channel can charge different fees for transactions in the two opposite directions in order to incentivize balanced flows. Nodes (users) are likely to make the rational choice of choosing the cheapest path. This automatically leads to routing decisions that reduce congestion and channel imbalance  \cite{engelmann2017towards, stasi2018routing, ren2018optimal, van2021merchant, varma2021throughput, wang2024fence}.
    \item There is a potential privacy concern that such channel prices can reflect channel balances and thereby leak sensitive transaction information. This concern can be offset by reducing the frequency of price updates or adding some noise to the price (differential privacy style solutions) \cite{tang2020privacy, wang2024fence}.
    \item Choosing the best path by minimizing the transaction fee over all possible paths in the graph is computationally expensive \cite{engelmann2017towards, stasi2018routing, chen2022routing}. Repeating this procedure for every transaction may be infeasible. A practical approach is for nodes to sporadically establish a small set of candidate paths among all possible ones and then optimize the prices only among these paths for every transaction \cite{van2021merchant}.
\end{itemize}

Next, we discuss the few papers that design PCN routing protocols with theoretical guarantees: the Spider protocol \cite{sivaraman2020high}, the protocol of \cite{varma2021throughput}, and Fence \cite{wang2024fence}. These protocols share some common features. First, they aim to reduce congestion and imbalance in PCN channels. Second, they involve channels signaling congestion or imbalance to transacting nodes (via prices or delays), who adjust routing decisions accordingly. Third, their theoretical guarantees stem from connections to network utility maximization problems.

Spider \cite{sivaraman2020high}, inspired by TCP protocols, breaks transactions into packets handled independently and maintains queues at both ends of each channel. A channel queues a packet if it lacks balance to forward it immediately, with large queues signaling congestion/imbalance and leading to delays. Nodes adjust transaction rates based on these delays, reducing rates when delays are high. The efficacy of Spider is demonstrated through extensive simulations. Moreover, \cite{sivaraman2020high} proves that its steady-state flow-rate optimizes a constrained utility maximization problem with logarithmic utilities.  However, this result is shown only for a specific network topology and arbitrarily large demand. Moreover, \cite{sivaraman2020high} does not provide any convergence guarantees.

Inspired by Spider, \cite{varma2021throughput} proposes a protocol that also breaks transactions into packets and maintains queues. However, it replaces delay-based signaling with explicit congestion and imbalance prices. Congestion prices are identical in both directions, while imbalance prices are opposite in sign. The path price is the sum of these two prices over all channels, and each packet takes the lowest-price path. The paper analyzes the protocol under i.i.d. demand and proves that under certain constraints on the mean demand, it stabilizes all queues (executes all transactions) with minimal on-chain rebalancing. The proof technique involves drawing parallels between the protocol and the dual to a constrained throughput maximization problem. However, the results hold only for a circulation demand (see discussion below).

Fence \cite{wang2024fence} also uses prices to manage congestion and imbalance and is based on weighted throughput maximization with constraints. Unlike the long-term flow-based models of \cite{sivaraman2020high} and \cite{varma2021throughput}, Fence adopts a finite-horizon approach, leading to a competitive ratio bound for worst-case demand. However, this result holds only for uni-directional PCNs, and \cite{wang2024fence} shows that no meaningful competitive ratio exists for bidirectional channels. Like \cite{varma2021throughput}, Fence tracks congestion and imbalance but uses an exponential price function instead of a linear one. These steep prices enable the competitive-ratio analysis in \cite{wang2024fence}.

We conclude this section with a short note on the fundamental limitations of any payment channel network in serving demand. Consider a setting where the demand between any two nodes is steady over time and the demand is small enough that capacity constraints are easily met (there is no congestion). The demand is said to be a  \textit{circulation} if each node sends as much money as it receives \cite{sivaraman2020high}. If this is not the case, \textit{i.e.}, some nodes may send more than they receive and vice-versa, the demand is said to be \textit{acyclic} \cite{sivaraman2020high}. Any demand can be split into a circulation component and a leftover acyclic component. In \cite{sivaraman2020high}, it is shown empirically that if the demand is a circulation, a PCN can support the entire demand perpetually without any on-chain rebalancing. Later, \cite{varma2021throughput} analyzes a very similar protocol and theoretically proves this result. In contrast, acyclic demands can be executed in full only with persistent on-chain rebalancing. In summary, a PCN may not always be able to serve the demand in entirety.

The papers \cite{sivaraman2020high} and \cite{sivaraman2021effect} also uncover the phenomenon of \textit{deadlocks}. A deadlock is a scenario where, owing to a few channels being imbalanced, a large number of transactions are rendered infeasible; moreover, they remain infeasible unless the channel undergoes on-chain rebalancing. \cite{sivaraman2021effect} explains that deadlocks may arise when the demand has an acyclic component. In particular, if a PCN tries to serve any transaction request as long as it is feasible, channels can get completely imbalanced, leading to deadlocks.  \cite{sivaraman2021effect} proposes a method to design the topology of the PCN to reduce the prevalence of deadlocks. However, it does not propose a protocol to avoid deadlocks. Preventing deadlocks is important for a PCN to operate efficiently in the long run. 

\section{Our Contribution}\label{sec:contribution}

In this work, we propose a simple routing and flow-control protocol for PCNs, called DEBT control, with provable convergence and optimality guarantees. We provide a brief overview of our protocol's design and analysis here in order to distinguish it from prior work. We assume we are given a payment channel network with an arbitrary topology; this topology is known and fixed. We are also given a set of transacting user-pairs who have some steady, elastic demand between them. Given the network and the demand, we set the objective of the PCN to maximize the total utility of all users, subject to the balance constraint imposed by each channel. The optimal solution to this problem is a set of flows that can be sustained indefinitely by the PCN without any rebalancing. We show that under the DEBT Control protocol, the flows in the network converge to some such optimal solution.

We arrive at the DEBT control protocol using the following methodology, which is inspired by similar works in the domain of communication networks \cite{kelly1998rate, low1999optimization, srikant2013communication, kelly2014stochastic}. We derive the dual of the aforementioned network utility maximization problem and show that the dual is an unconstrained convex optimization problem. We show that the gradient of the dual can be easily calculated, thus making it straightforward to use gradient-descent to solve the dual problem. Next, we show that the gradient-descent method admits a decentralized network implementation; this forms the DEBT control protocol. Endowing the Lagrange multipliers with the interpretation of channel prices gives an intuitive understanding of the protocol. Crucially, we allow channel prices to be negative. Using known results of strong duality (and other standard results in optimization), we show that the protocol converges to an optimal solution for the dual, and consequently, to optimal flows for the original problem. Our theoretical result, as well as the method to obtain it, is novel in the context of payment channel networks. 

Next, we compare our work to the most closely related prior work \cite{sivaraman2020high, varma2021throughput, wang2024fence}. Firstly, our protocol is designed from a long-term, or infinite-horizon point of view. This long-term viewpoint allows us to pose the network utility maximization problem in terms of flows, as done by \cite{sivaraman2020high}, and focus on the asymptotic convergence of the protocol to a steady-state solution. In contrast, the work of \cite{wang2024fence} takes a short-term, or finite-horizon, viewpoint. This difference is reflected in the flavor of the results: while \cite{wang2024fence} prove competitive ratio results, thereby giving guarantees for the worst-case performance over any possible demand, we prove the convergence of our protocol to the optimum flows under steady demands. Our model and result are also different from \cite{varma2021throughput}; it assumes stochastic demands, under which it proves bounds on the queue lengths, whereas we assume deterministic demands, and prove convergence of flows to the optimum value. Finally, our result is more general than the result of \cite{sivaraman2020high} in two ways. First, we show that our protocol converges to a steady-state solution, whereas \cite{sivaraman2020high} does not show any convergence guarantees. Second, our optimality guarantees hold for any network topology, any demand, and a wide class of utility functions, whereas the guarantees in \cite{sivaraman2020high} hold only for parallel networks, infinite demand, and logarithmic utility functions.

Our work also highlights the efficacy of price-based flow-control in preventing deadlocks. As we illustrate in Section \ref{sec:flow_control_example}, our protocol can selectively curb the acyclic component of the demand while allowing the circulation component of the demand to flow. This is indeed necessary for the flows to converge to an optimal steady-state value. It was observed in \cite{sivaraman2020high} that the Spider protocol is unsuccessful in preventing all deadlocks. Our analysis offers an explanation for this observation; we show that using logarithmic utilities does not allow the protocol to completely stifle deadlock-causing acyclic demands. See Section \ref{sec:discussion}A.c. 

\section{A Discrete-Time Model of a Payment Channel Network}\label{sec:PCN_model}
In this section, we present a mathematical model of a payment channel network. We follow the notation used by \cite{sivaraman2021effect}. The network consists of a set of nodes $V$ and a set of channels $E$ between pairs of nodes. The nodes are numbered $1, 2, \ldots, |V|$. A channel connecting nodes $u$ and $v$ is denoted $(u,v)$. We use the convention that the lower index vertex is written first; e.g., in the channel $(u,v)$, $u < v$. Each channel has a certain capacity, which refers to the total amount of money escrowed in the channel. Let $c_{u,v}$ denote the capacity of channel $(u,v)$ and let $c \in \mathbb{R}_+^{E}$ be a vector denoting the capacities of all the channels in the network. Thus, the tuple ($V, E, c$) specifies a weighted, undirected graph. This graph remains unchanged throughout the period of operation.

A PCN is a dynamical system. In this work, we assume the system evolves in discrete time steps. The state of the PCN at any given time is described by the balances in each of the channels. At any given time $t$, let the balance of node $u$ in channel $(u,v)$ be $x_{u,v}[t]$. It follows that the balance of $v$ in the same channel is $c_{u,v} - x_{u,v}[t]$. Let $x[t] \in \mathbb{R}_+^E$ denote the vector of balances, also called the state vector. By convention, $x$ contains the balance of the smaller-indexed node of each channel; the balance at the opposite end of the channel is inferred from its capacity. The state vector always satisfies $0 \leq x[t] \leq c$ (the inequalities hold component-wise).

\subsection{The Nature of Transaction Requests}\label{sec:transaction_requests}
\newcommand{\demand}{a}
\newcommand{\totalflow}{q}
We assume that in each slot, between every source-destination node-pair, a single transaction request arrives. Let $\demand_{i,j}[t]$ denote the monetary value of the transaction request from source $i$ to destination $j$ in slot $t$. This assumption does not sacrifice any generality. Indeed, the case where no transaction is requested is easily modeled by setting $\demand_{i,j}[t] = 0$. Moreover, multiple transaction requests in a time slot can be viewed as a single transaction request whose amount equals the sum of all the individual components. In general, the requested transaction amounts could vary arbitrarily over time. For the sake of simplicity, in this work, we focus on the regime of \textit{constant demands}, \textit{i.e.,} we assume that $\demand_{i,j}[t] = \demand_{i,j}$ for all time $t$ and all transacting node-pairs $(i,j)$. The vector $\demand = (\demand_{i,j})_{(i,j) \in V \times V}$ is called the \textit{demand vector}. We do not make any further assumption about $a$.

A second assumption we make is that the transaction demand arriving to the payment channel network is elastic. In other words, node-pairs prefer to have the entire transaction be served by the PCN, but it is acceptable that the request is dropped or partially served. This is a realistic assumption because users have alternate means of transacting, \textit{e.g.}, on the main blockchain. The choice of the transaction amount, $\totalflow_{i,j}[t]$, in response to the network state, is termed as \textit{flow-control}.

We model elastic transaction requests by means of a utility function. We assume that the node-pair $(i,j)$ gains a utility of $U_{i,j}(\totalflow_{i,j}[t])$ upon being served a transaction of amount $\totalflow_{i,j}[t] \in [0, a_{i,j}]$ by the network. We assume that $U_{i,j}(\cdot)$ is a concave, differentiable, and nondecreasing function over $[0, a_{i,j}]$. We also assume that $U_{i,j}(0) = 0$ and $U'_{i,j}(0) < \infty$. We assume transaction requests are \textit{infinitely divisible}, which means that any fraction of the value $\demand_{i,j}[t]$ can be served in slot $t$ with the rest being dropped. We also assume that transactions are not queued. 

\subsection{The Order of Events}
We model the PCN as a discrete-time system. To be definite, we assume that events take place in a fixed order in each time slot, as described below.
\begin{enumerate}
    \item Transaction requests arrive to the payment channel network. Each transaction request is composed of a source-destination node-pair and a monetary value. These transaction requests depend on exogenous factors, such as demand for certain goods and services, and are independent of the state in the network. 
    \item The nodes make flow-control and routing decisions. That is, for every transaction request, the corresponding source-destination pair decides what fraction of the transaction value should be served; this includes the possibilities of serving the transaction entirely or dropping the transaction entirely. The node pair also decides which paths of the network should carry the transaction--we allow for multi-path routing. These decisions are made on the basis of prices that reflect the network's state; the precise details are given in Section \ref{sec:debt_control_protocol}.
    \item The payment channel network executes the transaction requests made by the nodes by moving the requisite money through the payment channels. After the transactions are served, channel balances are updated. We assume that the channels themselves never drop any transactions, \textit{i.e.}, all transaction requests that are passed on from the nodes to the network are actually served. The details of how this may be achieved in practice is described in Section \ref{sec:flow_feasibility}. 
\end{enumerate}

\subsection{Paths and Flows} A path in the PCN is a sequence of channels, with each one adjacent to the previous one, endowed with a sense of direction. A path may have cycles, but we assume that it traverses each channel at most once. Any such path can be represented by a vector $r \in \{-1, 0, 1\}^{E}$ using the convention described below:
\begin{itemize}
    \item let $r_{u,v} = 1$ if the path traverses the channel $(u,v)$ in the direction $u \rightarrow v$,
    \item let $r_{u,v} = -1$ if the path traverses the channel $(u,v)$ in the direction $v \rightarrow u$, and
     \item let $r_{u,v} = 0$ if the path does not traverse the channel $(u,v)$ in either direction.
\end{itemize}

A source-destination pair $(i,j)$ may use any number of paths in the network for carrying transactions. Denote the $k\textsuperscript{th}$ path by $p_{i, j, k}$ and the set of such paths by $P_{i,j}$. Note that a path from $i$ to $j$ is different from a path from $j$ to $i$. Let $P$ denote the set of all paths ($P = \cup_{i,j} P_{i,j}$). Let $R$ denote the $E \times P$ routing matrix with entries in $\{-1, 0, 1\}$ constructed using the convention given above. In $R$, each column corresponds to a particular path and each row corresponds to a particular channel.

With every path in the network, we associate a flow, which represents the amount of money sent along that path over a period of time. Let $f_{i,j,k}[t]$ denote the amount of money being sent on path $p_{i,j,k}$ in slot $t$. The amounts of money sent from node $i$ to node $j$ along all possible paths in slot $t$ is denoted by $f_{i,j}[t]$, \textit{i.e.}, \(f_{i,j}[t] \triangleq (f_{i,j,k}[t] : 1\leq k \leq |P_{i,j}|) \).  The total amount of money sent from $i$ to $j$ in slot $t$ is denoted by $\totalflow_{i,j}[t]$. Thus \(\totalflow_{i,j}[t] \triangleq \Sigma_k f_{i,j,k}[t].\) Finally, let $f[t] \in \mathbb{R}^{|P|}$ denote the vector of all the flows in the network in slot $t$.

\subsection{Feasibility of Flows and State-Change Equations}\label{sec:flow_feasibility} 
If the channel balances are arbitrary, it is possible that some channels are unable to execute the requested flows through them due to insufficient balance. A set of flows can be served entirely if and only if there is sufficient balance on each side of each edge to route all the flows through in that direction simultaneously. Given a balance vector $x[t]$, a flow vector $f[t]$ is \textit{feasible} if:
\begin{equation}\label{eq:feasibility}
    R^+f[t] \leq x[t]; \quad R^-f[t] \leq c - x[t].
\end{equation}
Here, $R^+$ is the matrix obtained by turning all $-1$s to $0$s in $R$ and $R^-$ is the matrix obtained by turning all $1$s to $0$s and $-1$s to $1$s in $R$. Thus, $R = R^+ - R^-$.

In this work, we make the following two assumptions that ensure that the flow vector at each time step is always feasible. Firstly, the channel capacities must be large enough so that if the channels are evenly balanced, they will be able to serve the entire demand on any routing scheme. We make this notion precise as follows. Given any demand vector $a$, let $\mathcal{N}$ denote the set of transacting node-pairs, \textit{i.e.}, $\mathcal{N} = \{(i,j): a_{i,j} > 0\}$ and let $A$ denote the set of (non-negative) flows satisfying the demand constraints:
\begin{align}
    A \triangleq \{f\in R^{|P|}: f \geq 0, \ \sum_k f_{i,j,k} \leq a_{i,j} \ \forall \ (i,j) \in \mathcal{N}\}.
\end{align}
We assume that the capacities of the channel are large enough such that:
\begin{align}\label{eq:large_capacity_assumption}
    \sup_{f \in A} R^+f \leq c/2, \quad \sup_{f \in A} R^-f \leq c/2
\end{align}
This and other assumptions are discussed in Section \ref{sec:discussion}.

Secondly, whenever a channel observes that the requested flow through it is infeasible, it rebalances itself and then executes the flow. (By the first assumption, the flow becomes feasible after this rebalancing operation.) To keep things simple, we assume that the rebalancing operation takes place instantaneously. In practice, a rebalancing operation can take time; however, a channel can preemptively rebalance itself if it observes that its balances are getting skewed beyond a certain extent.

With this notation in place, we can write the state-update equations as follows. First, the channels undergo rebalancing if necessary. Thus, the intermediate balances are:
\begin{equation}\label{eq:rebalancing}
    \Tilde{x}_{u,v}[t] = \begin{cases}
        c_{u,v}/2 & \text{if  channel $(u,v)$ rebalances} \\
        x_{u,v}[t] & \text{otherwise}
    \end{cases}
\end{equation}
Next, they execute the requested flow. At the end of the slot, the new balances are:
\begin{equation}\label{eq:balance_update}
    x[t+1] = \Tilde{x}[t] - Rf[t].
\end{equation}

\subsection{Detailed Balance Flows}\label{sec:detailed_balance_flows}
The notion of flows as described above is a quantity that is dynamically evolving in time. In contrast, the notion of a \textit{stationary flow} refers to a steady, long-term, exchange of transactions between a pair of nodes in the PCN. (This is similar to the usage in Section \ref{sec:payment_channels}.)  
When referring to stationary flows, we use the notation $f$ (without $t$) to denote the vector of flow rates over each path. We also use the same convention to define $f_{i,j}$ and $\totalflow_{i,j}$.

In Section \ref{sec:payment_channels}, we discussed the balanced flow constraint on a payment channel. This constraint applies to every channel in a payment channel network. A stationary flow vector $f$ is said to satisfy the \textit{detailed balance condition} (equivalently, said to be a detailed balance flow) if \(Rf = 0.\) For any flow vector $f$, the term $(Rf)_{u,v}$ represents the net flow in the direction $u \rightarrow v$ along the channel $(u,v)$; a negative value means that the net flow is in the direction $v \rightarrow u$. A detailed balance flow is such that the amount of money flowing through each channel is equal in the two opposite directions. Thus, after a detailed balance flow has been served, the balances on all the edges remain the same as before (see \eqref{eq:balance_update}). Any detailed balance flow can be sustained in steady-state without any on-chain rebalancing. Conversely, any flow vector that does not satisfy the detailed balance condition can only be sustained in steady-state with persistent on-chain rebalancing. 

\section{The DEBT (DEtailed Balance Transactions) Control Protocol for PCNs}\label{sec:debt_control_protocol}
\subsection{Overview}\label{sec:protocol_overview}
The model described in the previous section provides a framework for discussing network protocols for PCNs. In this section, we present such a protocol which, under any stationary demand, asymptotically maximizes the total utility of all transacting node-pairs while avoiding the expensive operation of periodic on-chain rebalancing. The protocol guides the transacting node-pairs to make flow-control and routing decisions such that the flows in the network ultimately converge to a suitable stationary flow $f^*$. On the one hand, $f^*$ should maximize the sum of all node-pairs' utilities; here, we extend the notion of utility from individual transactions to transaction rates. On the other hand, $f^*$ must be a detailed balance flow and should not exceed the requested demand. We call this protocol the DEBT (DEtailed Balance Transactions) control protocol. 

The key idea of the protocol is to use channel prices as a mechanism to control flows. To elaborate, channels quote prices to nodes for routing flows through them. These channel prices are directional; the price in one direction is always the negative of the other. The price of a path is the sum of the prices of the channels along the path. If multiple paths exist between the same pair of nodes, the protocol recommends choosing a path with the least price or splits the flow along multiple competitive paths. The flow-control decisions are made by comparing the utility gained in having the transaction served to the cost for serving a transaction along a path, where the cost is the product of the price and the transaction volume. A high path price signals the nodes to reduce the flow along the path, while a low path price provides an incentive to increase the corresponding flow.

The DEBT control protocol is an iterative one. Initially, all channel prices are zero; consequently, the path prices are zero as well. Therefore, all transaction requests are served and routing choices are made arbitrarily. Over time, channels adjust their prices based on the net flow through them. A channel $(u,v)$ increases the price in the direction $u \rightarrow v$ if there is a net flow in that direction. This change in the price dissuades further flow in that direction and encourages flow in the opposite direction. As the flows converge to a detailed balance flow, the channel prices begin to converge as well. 

Prior to convergence, the flows along each channel may not be balanced; indeed, without an imbalanced flow, prices would not change at all. An imbalanced flow causes the channel balances to change over time. If a channel's capacity is large enough, the cumulative net flow of transactions throughout the transient period never depletes the channel's balances at either end. In this case, the channel will never undergo rebalancing. Else, a channel might occasionally find it infeasible to route the requested flows due to its skewed balances. Whenever such an occasion arises, a channel undergoes on-chain rebalancing and resets its balances. Eventually, no more rebalancing is required as the flows converge to the flow $f^*$ satisfying the detailed balance condition.

The DEBT control protocol is derived in the following steps:
\begin{enumerate}
    \item The objective of the protocol is posed as a \textit{network utility maximization} problem, which we call the primal problem. An optimal set of flows is one which maximizes the sum of all the node-pair's utilities, subject to the constraint that flows satisfy the detailed balance condition and the flow rates served do not exceed the desired flow rates.
    \item The dual problem corresponding to the primal problem is derived by introducing Lagrange multipliers for the detailed balanced constraint. By the theory of strong duality for convex programs, a solution to the dual problem provides a solution to the primal one.
    \item An iterative gradient-descent algorithm is proposed to solve the dual problem. The algorithm updates the Lagrange multipliers (dual variables) in small steps and the flows (primal variables) are set in response.
    \item The above algorithm is shown to have a decentralized implementation, suitable for implementation on a PCN. The Lagrange multipliers are interpreted as prices quoted by channels. This decentralized implementation constitutes the DEBT control protocol.
\end{enumerate}

\subsection{A Network Utility Maximization Problem}
To formulate the protocol's objective in mathematical terms, we introduce some notation. Let $U(f) = \sum_{(i,j) \in \mathcal{N}} U_{i,j}(\totalflow_{i,j})$ denote the total utility of all transacting node-pairs as a function of a stationary flow $f$. Define a \textit{feasible flow} to be any flow that meets both the demand constraints and the detailed balance constraints (\textit{i.e.}, the condition $Rf = 0$).

As we will see later, it will be important in our proofs that the flows in the network respond smoothly to the prices. For this reason, we add a quadratic regularizer term, $H(f)$, to the utility function, where
\begin{equation}\label{eq:entropy_regularizer}
    H(f) \triangleq -\sum\nolimits_{(i,j) \in \mathcal{N}}\eta_{i,j} \sum\nolimits_{k=1}^{|P_{i,j}|} (f_{i,j,k})^2.
\end{equation}
Here, $\eta_{i,j}$ is a non-negative scalar term that controls the weight of the regularizer. One interpretation of the regularizer is that it is an incentive for splitting the flow among different competitive paths (see Section \ref{sec:prices} for details). Define $U(f) + H(f)$ to be the \textit{net utility} of a particular flow $f$.

The protocol's goal is to find a feasible stationary flow that maximizes the net utility of the payment channel network. In mathematical terms, this can be expressed as obtaining a solution to the following optimization problem:
\begin{equation}\label{eq:primal_problem}
\tag{$\mathbf{P}$}
\begin{split}
    \max_{f \in A} \quad & U(f) + H(f)\\ 
    \text{s.t.} \quad &Rf = 0 
\end{split}
\end{equation}
The symbol $\eqref{eq:primal_problem}$ denotes that the optimization problem presented above is the \textit{primal} (or original) problem. Let $f^*$ denote any solution to this problem.

Observe that the set of feasible flows is a compact, convex set. Moreover, it is nonempty for any problem parameters, since the empty flow ($f = 0$) is a feasible flow. Therefore, a solution to $\eqref{eq:primal_problem}$ always exists. Also note that \eqref{eq:primal_problem} is a convex optimization problem, since both $U(f)$ and $H(f)$ are concave and the constraint set is convex. Lastly, if all $\eta_{i,j}$ are strictly positive, then the objective function is strongly concave. This ensures that $f^*$ is unique.

\subsection{The Dual Problem}\label{sec:dual_problem}
The primal problem \eqref{eq:primal_problem} is hard to solve because of the detailed balance constraints. With the aid of Lagrange multipliers, we derive its \textit{dual problem} that does not explicitly have these constraints. (See \cite{bertsekas1999nonlinear} for an exposition on Lagrange multipliers and duality).
Let $\lambda_{u,v}$ denote the Lagrange multiplier for the constraint $(Rf)_{u,v} = 0$; let $\lambda \in \mathbb{R}^E$ denote the vector of all such terms. Define the Lagrangian of the problem $\eqref{eq:primal_problem}$ by
\begin{align}\label{eq:lagrangian}
    L(f, \lambda) \triangleq U(f) + H(f) - \lambda^TRf.
\end{align}

Using the Lagrangian, we can formulate an equivalent form of the problem \eqref{eq:primal_problem} as follows:
\begin{align}\label{eq:lagrangian_max_min}
    &\max_{f \in A} \  \inf_{\lambda \in \mathbb{R}^E}\quad L(f, \lambda) \nonumber \\ 
    = &\max_{f \in A} \  \inf_{\lambda \in \mathbb{R}^E}\quad  U(f) + H(f) - \lambda^TRf
\end{align}
Observe that \eqref{eq:lagrangian_max_min} is equivalent to \eqref{eq:primal_problem} because 
$f$ must be chosen such that $Rf = 0$ if \eqref{eq:lagrangian_max_min} is to have a finite value. 

The dual of the optimization problem \eqref{eq:primal_problem} 
is obtained by changing the order of minimization and maximization in \eqref{eq:lagrangian_max_min}:
\begin{align}\label{eq:dual_problem}
    &\inf_{\lambda \in \mathbb{R}^E} \  \max_{f \in A}\quad L(f, \lambda) 
    = \inf_{\lambda \in \mathbb{R}^E} D(\lambda), \tag{$\mathbf{D}$}
\end{align}
where $D(\lambda)$, called the \textit{dual function}, is defined as follows:
\begin{align}\label{eq:dual_function}
    D(\lambda) \triangleq \max_{f \in A} U(f) + H(f) - \lambda^TRf \ \  \forall \, \lambda \, \in \, \mathbb{R}^{E}.
\end{align}
For any $\lambda$, $L(f, \lambda)$ is finite for all $f \in A$, because $A$ is a bounded set. Therefore, $D(\lambda)$ is well-defined for all $\lambda \in \mathbb{R}^{E}$. 

Observe that the dual function is a convex function of $\lambda$ \cite{bertsekas1999nonlinear} and that the dual problem has no constraints. This implies that the dual problem is easy to solve. In Section \ref{sec:convergence}, we show how a solution of \eqref{eq:dual_problem} yields a solution of \eqref{eq:primal_problem}.

\subsection{A Dual Algorithm}\label{sec:dual_algorithm}

The gradient descent method is a classical method to solve unconstrained convex optimization problems. To use this method to solve \eqref{eq:dual_problem}, we need to establish conditions under which $D(\lambda)$ is differentiable and also obtain an expression for the gradient of $D(\lambda)$. Lemma \ref{lem:dual_subgradient} gives us an expression for the subdifferential of $D(\lambda)$. Because $D(\lambda)$ is a convex function, the subdifferential set is nonempty at all points. $D(\lambda)$ is differentiable precisely at those points where the subdifferential set has a unique element. The lemma follows immediately from Danskin's theorem, also known as the envelope theorem. (See Appendix B of \cite{bertsekas1999nonlinear} for the precise statement and proof of Danskin's theorem).

\begin{lemma}\label{lem:dual_subgradient} 
Let $D(\lambda)$ be the function as defined in \eqref{eq:dual_function}. The subdifferential set of $D(\lambda)$ is given by
\[\partial D(\lambda) = \left\{ \nabla_\lambda L(f, \lambda): f \in F(\lambda) \right\} = \left\{ -Rf: f \in F(\lambda) \right\}\]
where $F(\lambda) \triangleq \arg \max_{f \in A} L(f, \lambda)$ is the set of all flow vectors that maximize the Lagrangian, given $\lambda$.
\end{lemma}

In Section \ref{sec:convergence}, we show that as long as the regularizer coefficient $\eta_{i,j}$ is strictly positive for all $(i,j) \in \mathcal{N}$, $F(\lambda)$ has a unique value for all $\lambda \in \mathbb{R}^E$. By Lemma \ref{lem:dual_subgradient}, $D(\lambda)$ is differentiable everywhere whenever this condition holds.

The gradient descent algorithm to solve the dual problem is presented below. Initialize the algorithm by setting $\lambda[0]$ to be the zero vector. For every $t \in \mathbb{N}$, set
\begin{equation}\label{eq:algorithm}\tag{$\mathbf{A}$}
\begin{split}
    f[t] &= \arg \max_{f \in A} L(f, \lambda[t]) \\
    \lambda[t+1] &= \lambda[t] + \gamma Rf[t]
\end{split}
\end{equation}
Here, $\gamma$ is a strictly positive stepsize parameter in the algorithm that remains constant for all time. In each iteration $t$, the flows $f[t]$ are set so as to maximize the Lagrangian, given the current values of $\lambda [t]$, while the Lagrange multipliers $\lambda[t+1]$ are updated in a direction opposite to the gradient of $D(\lambda[t])$. In case there is more than one value of $f$ that maximizes $L(f, \lambda[t])$, we can set $f[t]$ to any such value. In this case, algorithm \eqref{eq:algorithm} is equivalent to the subgradient method applied to $D(\lambda)$. Under appropriate conditions, we expect $\lambda[t]$ to approach a solution of \eqref{eq:dual_problem}, and consequently, $f[t]$ to approach a solution of \eqref{eq:primal_problem}. We defer the convergence analysis of \eqref{eq:algorithm} to Section \ref{sec:convergence}. For now, we explain how \eqref{eq:algorithm} can be implemented as a network protocol in a PCN.

\subsection{Lagrange Multipliers as Prices}\label{sec:prices}
The first step towards an intuitive interpretation of \eqref{eq:algorithm} is to observe that the Lagrangian is a sum of terms, each concerning one transacting node-pair. Observe that the Lagrangian depends on $\lambda$ only through the term $R^T \lambda$ (see \eqref{eq:lagrangian}). Define $\mu \triangleq R^T\lambda$. Then $\mu$ is a vector indexed by the paths in the PCN, such that
\begin{equation}\label{eq:path_price}
    \mu_{i,j,k} = \sum\nolimits_{u \rightarrow v \in p_{i,j,k}} \lambda_{u,v} - \sum\nolimits_{v \rightarrow u \in p_{i,j,k}} \lambda_{u,v}.
\end{equation}
Similar to the notation $f_{i,j}$, define $\mu_{i,j} \triangleq (\mu_{i,j,1}, \ldots, \mu_{i,j,k})$. Further, define $\Tilde{L}(f, \mu)$ to be 
\begin{align}\label{eq:lagrangian_f_mu}
    &\Tilde{L}(f, \mu) \triangleq \sum\nolimits_{(i,j) \in \mathcal{N}} L_{i,j}(f_{i,j}, \mu_{i,j}), \text{ where } \\
    &L_{i,j}(f_{i,j}, \mu_{i,j}) \triangleq U_{i,j}(\totalflow_{i,j}) - \sum_{k = 1}^ {|P_{i,j}|}\left(f_{i,j,k}\mu_{i,j,k} \,+ \eta_{i,j}(f_{i,j,k})^2\right) \nonumber
\end{align}

Next, observe that in \eqref{eq:algorithm}, the flows are chosen by solving:
\[ f[t] = \argmax_{f \in A} L(f, \lambda[t])  = \Tilde{L}(f, \mu[t]) \, ; \quad \mu[t] = R^T\lambda[t].\]
The expression of $\Tilde{L}(f, \mu)$ given in \eqref{eq:lagrangian_f_mu} shows that this optimization problem also separates across node-pairs. Therefore, given $\mu$, the flows between each $(i,j) \in \mathcal{N}$ can be determined independently by solving:
\begin{align}
\label{eq:flow_opt1}
    f_{i,j}[t] = \argmax_{\{f_{i,j}: \, f_{i,j,k} \geq 0 \, \forall \, k, \ q_{i,j} \leq a_{i,j}\}} \  L_{i,j}(f_{i,j}, \mu_{i,j}[t])
\end{align}

The next step towards interpreting algorithm \eqref{eq:algorithm} is to endow the Lagrange multipliers with some meaning. 
Let $\lambda_{u,v}$ be interpreted as the \textit{channel price}, i.e., the cost of routing one unit of flow in the direction $u \rightarrow v$ through the channel $(u,v)$. Let the price for routing flows in the opposite direction of the same channel be $-\lambda_{u,v}$. Then $\mu_{i,j,k}$ can be interpreted as the \textit{path price}, i.e., the cost that the node pair $(i,j)$ needs to pay to send a unit flow along the path $p_{i,j,k}$. The path price is equal to the sum of the channel prices of each of the channels in the path with the appropriate direction (or sign) incorporated, as shown in \eqref{eq:path_price}. Note that the channel prices, and therefore the path prices, may be negative. 

With this interpretation, $f_{i,j,k} \mu_{i,j,k}$ is the cost of sending a flow of amount $f_{i,j,k}$ along the path $p_{i,j,k}$. Thus, $\sum_k f_{i,j,k} \mu_{i,j,k}$ is the total cost incurred by the node-pair $(i,j)$ for splitting the total flow amount $\totalflow_{i,j}$ along different paths. In  $L_{i,j}(f_{i,j}, \mu_{i,j})$, this cost is subtracted from the utility gained by the node-pair $(i,j)$ in executing a transaction of amount $\totalflow_{i,j}$  (see \eqref{eq:lagrangian_f_mu}). The quadratic term $\eta_{i,j} \sum_{k} (f_{i,j,k})^2$, coming from the regularizer, can be interpreted as a penalty for concentrating all flows along a single path. Equivalently, it acts as an incentive to split the total flow along different paths. This is because for any fixed value of $\totalflow_{i,j}$, the sum $ \sum_{k} (f_{i,j,k})^2$ is minimized by splitting $\totalflow_{i,j}$ equally among all $f_{i,j,k}$. Thus, the interpretation of the optimization problem in \eqref{eq:flow_opt1} is that each node-pair tries to maximize its net utility, i.e., the utility of executing a transaction with the cost of execution subtracted from it. This is a rational decision for the nodes.

\subsection{Flow-Control and Routing Decisions} \label{sec:flow_control_routing} 
The previous section establishes that each node-pair can solve \eqref{eq:flow_opt1} in parallel, independently of each other. We now show how solving \eqref{eq:flow_opt1} can be interpreted as simultaneously making routing and flow-control decisions. First, consider the case when $\eta_{i,j}$ is zero. In this case, it is optimal to route the transaction only along the path with the minimum price; all other paths from $i$ to $j$ carry zero flow. Let $\mu^*_{i,j}[t]$ denote the minimum path price. The amount of flow carried by this path, $\totalflow_{i,j}[t]$, is given by: 
\begin{equation}\label{eq:scalar_flow}
    \totalflow_{i,j}[t] = \argmax_{\totalflow \in [0, a_{i,j}]}U_{i,j}(\totalflow) - \totalflow \mu^*_{i,j}[t].
\end{equation}
The choice of the total amount of flow is interpreted as a flow-control action and the choice of the path to carry the flow is interpreted as a routing decision.

Now consider the case where the regularizer coefficient $\eta_{i,j}$ is strictly positive. When there is a single path between $i$ and $j$, the quadratic regularizer term can be absorbed into the utility function. Thus, the problem in \eqref{eq:flow_opt1} reduces to \eqref{eq:scalar_flow}. When there are multiple paths from $i$ to $j$, the solution to \eqref{eq:flow_opt1} can be expressed in terms of the classical \textit{waterfilling scheme}. To illustrate this, we invoke the idea of Lagrange multipliers once again to deal with the constraints in \eqref{eq:flow_opt1}. Let $\nu_{i,j}$ denote the Lagrange multiplier for the demand constraint $\totalflow_{i,j} \leq a_{i,j}$. By the KKT conditions \cite{bertsekas1999nonlinear}, the optimal solution to \eqref{eq:flow_opt1} must satisfy:
\begin{align}\label{eq:waterfilling}
f_{i,j,k} = \left(\frac{U'_{i,j}(\totalflow_{i,j}) - \nu_{i,j} - \mu_{i,j,k}}{2\eta_{i,j}}\right) ^+ \ \forall \ k.
\end{align}
The total flow $\totalflow_{i,j}$ must also satisfy the demand constraint ($\totalflow_{i,j} \leq a_{i,j}$) and the  complementary slackness condition: $(\totalflow_{i,j} - a_{i,j})\nu_{i,j} = 0$. Lastly, each $\nu_{i,j}$ must be nonnegative.

Interpreting \eqref{eq:waterfilling} as a waterfilling scheme is easiest when the utility function is linear and the optimal solution satisfies $\totalflow_{i,j} < a_{i,j}$. In this case, $\nu_{i,j}$ is equal to zero and $U'_{i,j}(\totalflow_{i,j})$ is constant; call it $U'_{i,j}$. The term $U'_{i,j}$ acts as a \textit{price ceiling}; any path with a price larger than this price ceiling does not carry any flow. The remaining paths carry a flow that is proportional to the gap between the path price and the price ceiling. Thus, the path with the least price gets allotted the maximum flow and the other paths get progressively smaller flows. Paths with the same price always get allotted the same amount of flow.

Now consider the general case. 
The method to solve \eqref{eq:waterfilling} is based on the following observation: increasing either $\totalflow_{i,j}$ or $\nu_{i,j}$ in \eqref{eq:waterfilling} tends to decrease the individual flows $f_{i,j,k}$. 
If all the path prices are bigger than $U'_{i,j}(0)$, then the optimal solution is to set all the flows equal to zero. If not, set $\totalflow_{i,j}$ to $\demand_{i,j}$ and $\nu_{i,j} = 0$ in \eqref{eq:waterfilling} and check whether the corresponding flows, $f_{i,j,k}$, add up to a value more than $\demand_{i,j}$. If so, keep $\totalflow_{i,j} = \demand_{i,j}$ and increase $\nu_{i,j}$ to the value such that the corresponding flows add up exactly to $\demand_{i,j}$. If not, keep $\nu_{i,j} = 0$ and find the appropriate value of $\totalflow_{i,j}$ using a few iterations of binary search as follows. For any $\totalflow_{i,j} \in (0, \demand_{i,j})$, check whether the flows given by \eqref{eq:waterfilling} add up to less than $\totalflow_{i,j}$ or not, and adjust $\totalflow_{i,j}$ accordingly. Once $\totalflow_{i,j}$ is fixed, the flows on the individual paths are given by \eqref{eq:waterfilling}.

\subsection{From a Dual Algorithm to the DEBT Control Protocol}\label{sec:algorithm_to_protocol}

We now describe how algorithm \eqref{eq:algorithm} can be implemented in a decentralized fashion in a payment channel network whose model we presented in Section \ref{sec:PCN_model}. This decentralized implementation is the DEBT control protocol for PCNs. {At every slot $t$:
\begin{enumerate}
    \item The channel prices, $\lambda[t]$, are made publicly available. Using these prices and the knowledge of the graph topology, transacting node pairs calculate the prices along their prospective paths according to \eqref{eq:path_price}. 
    \item The demand arrives to the network. Each transacting node pair $(i,j)$ knows its demand $a_{i,j}$ and its path prices $\mu_{i,j,k}[t]$ (from the previous step). Using these, each node pair calculates the flow $f_{i,j,k}[t]$ to be requested along each path, using the waterfilling scheme \eqref{eq:waterfilling}. In Section \ref{sec:flow_control_routing}, we noted that each node-pair can perform this calculation independently.
    \item The flow requests are conveyed to the respective channels. Each channel checks whether the flow requested through it is feasible or not (whether it has sufficient balance to execute the flow). If it does, it immediately executes the flow; if not, it rebalances itself and then executes the flow. In either case, the flow requests made at slot $t$ are executed within the same slot. Following this, the balances of the channels change according to \eqref{eq:rebalancing}-\eqref{eq:balance_update}.
    \item Finally, each channel updates its price proportional to the net flow through it:
    $$ \lambda_{u,v}[t+1] = \lambda_{u,v}[t] + \gamma (Rf)_{u,v}[t] \quad \forall \ (u,v) \in E $$
    The price updates are based on local information alone; thus, each channel updates its price independently. In fact, channels need not know the source or the destination of the flows that they are serving to calculate prices.
\end{enumerate}}

The term $\lambda_{u,v}$ tends to increase if the net flow through the channel $(u,v)$ is in the direction from $u \rightarrow v$, and tends to decrease if the net flow is in the opposite direction. This change is consistent with the interpretation of $\lambda_{u,v}$ as a price that penalizes or encourages flows in order to maintain detailed balance. Indeed, a sustained net flow in either direction is bound to increase the price for any future flow in that direction. It also decreases the price for any flow in the opposite direction, thereby encouraging such flows. Note that when a channel rebalances itself, it does not reset its prices. Thus, the price of a channel may not exactly reflect its balances; rather, it is proportional to the net amount of money that a channel has carried through itself. The prices keep adjusting in a manner such that eventually, the net flow through each channel converges to zero. In the following section, we show that the flows not only converge to a detailed balance flow, but in fact, they converge to a solution of \eqref{eq:primal_problem}.

\newcommand{\N}{{\mathcal{N}}}
\newcommand{\D}{{\mathcal{D}}}
\newcommand{\Hc}{{\mathcal{H}}}
\newcommand{\bmu}{{\boldsymbol\mu}}
\newcommand{\btheta}{{\boldsymbol\theta}}
\newcommand{\tlim}{{t \rightarrow \infty}}

\section{Convergence Analysis}\label{sec:convergence}

Our first result establishes that the dual problem always has a (finite) solution, and the optimal flow in response to such a solution is primal optimal.
\begin{lemma}\label{lem:strong_duality} For any instance of the primal problem \eqref{eq:primal_problem}, the corresponding dual problem \eqref{eq:dual_problem} has a solution, \textit{i.e.}, there exists $\lambda^* \in \mathbb{R}^E$ such that 
\(D(\lambda^*) = D^* \triangleq \inf_{\lambda \in \mathbb{R}^E} D(\lambda). \)
Further, the set $F(\lambda^*) = \argmax_{f \in A} L(f, \lambda^*)$ contains a solution to \eqref{eq:primal_problem}.
\end{lemma}
\begin{proof} The lemma follows immediately from Proposition 3.4.2 of \cite{bertsekas1999nonlinear}, which states sufficient conditions for strong duality to hold. Those conditions are satisfied by our problem because: 
\begin{itemize}
    \item \eqref{eq:primal_problem} has a solution $f^*$ because the set of feasible solutions is non-empty and compact.
    \item The constraint set $A$ is a convex polytope and the detailed balance condition is a linear constraint.
    \item The objective function $U(f) + H(f)$ can be extended for all $f \in \mathbb{R}^P$ in a continuously differentiable manner by defining each $U_{i,j}(\cdot)$ outside the interval $[0, a_{i,j}]$ to be appropriate affine functions.
\end{itemize}
By Proposition 3.4.2(a) of \cite{bertsekas1999nonlinear}, \eqref{eq:dual_problem} has a solution $\lambda^*$. Since $\lambda^*$ is a minimizer of $D(\lambda)$, it must be that $0 \in \partial D(\lambda^*)$. By Lemma \ref{lem:dual_subgradient}, this implies that there exists $f^* \in F(\lambda^*)$ such that $Rf^* = 0$, \textit{i.e.}, there exists a feasible $f^*$ in $F(\lambda^*)$. By Proposition 3.4.2(b), this $f^*$ is a solution to \eqref{eq:primal_problem}.
\end{proof}

Our next result is to establish conditions under which the dual function is smooth. We invoke standard properties of the Fenchel conjugate of convex functions to prove this result (see Section 2.7 of \cite{shalev2012online} for a reference).

\begin{assumption}\label{assumption_eta_positive}
    $\eta_{i,j} \geq \eta > 0 \ \forall \ (i,j) \in \N$. \newline
\end{assumption}

\begin{lemma}\label{lem:smoothness}
    Under Assumption \ref{assumption_eta_positive}:
    \begin{itemize}
        \item $F(\lambda) = \argmax_{f \in A} L(f, \lambda)$ is a singleton for all $\lambda$.
        \item The dual function is smooth with parameter $\Vert R \Vert_{op}^2/\eta$, where $\Vert \cdot \Vert_{op}$ denotes the operator norm of a matrix.
        \item $F(\lambda)$ is a continuous function of $\lambda$.
    \end{itemize}
\end{lemma}
\begin{proof}
    Under Assumption \ref{assumption_eta_positive}, the Lagrangian, $L(f, \lambda)$, is  $\eta$-strongly concave in $f$. This is because $L(f,\lambda)$ is the sum of an $\eta-$strongly concave function $H(f)$, a concave function $U(f)$, and a linear function $\lambda^TRf$. The strong concavity of $L$ implies the uniqueness of $F(\lambda)$.
    
    To show the smoothness of $D(\lambda)$ with respect to $\lambda$, define $\Tilde{D}(\mu) \triangleq \max_{f \in A} \Tilde{L}(f, -\mu)$, where $\Tilde{L}(f, \mu)$ is defined in \eqref{eq:lagrangian_f_mu}. By this construction, $\Tilde{D}(\mu)$ is the Fenchel conjugate of $-(H(f) + U(f))$. Since $-(H(f) + U(f))$ is an $\eta$-strongly convex function, $\Tilde{D}(\mu)$ is a $1/\eta$-smooth function of $\mu$ \cite{shalev2012online}. Since $D(\lambda) = \Tilde{D}(-R^T\lambda)$, it follows that $D(\lambda)$ is $\Vert R \Vert_{op}^2/\eta$-smooth function of $\lambda$.

    By Lemma \ref{lem:dual_subgradient}, $\nabla D(\lambda) = RF(\lambda)$. By definition of smoothness, $\nabla D(\lambda)$ is a continuous function of $\lambda$, which in turn, implies the continuity of $F(\lambda)$ with respect to $\lambda$.
\end{proof} 

We now present our main result, which states that with suitably small stepsizes, the DEBT control protocol converges to a solution of \eqref{eq:primal_problem}.

\begin{assumption}\label{assumption_small_stepsize}
    The stepsize of \eqref{eq:algorithm} satisfies $\gamma < \eta/\Vert R \Vert^2_{op}$.\newline
\end{assumption}

\begin{proposition}\label{prop:convergence}
    Under Assumptions \ref{assumption_eta_positive} and \ref{assumption_small_stepsize}:
    \begin{itemize}
        \item \(D(\lambda[t]) - D(\lambda^*) \leq \frac{\Vert \lambda^* \Vert}{2\gamma t}  \ \forall \ t \geq 1\).
        \item $\lambda[t] \rightarrow \lambda^{**}$ for some $\lambda^{**} \in \argmin_{\lambda \in \mathbb{R}^E} D(\lambda)$  as $t \rightarrow \infty$.
        \item {$f[t] \rightarrow f^*$ as $t \rightarrow \infty$}, where $f^*$ is the unique solution to \eqref{eq:primal_problem}.
    \end{itemize}
\end{proposition}
\begin{proof}
    The first point is a standard result in convex optimization, namely gradient descent with a constant stepsize, when applied to a smooth convex function, leads to the function value converging to its infimum at a rate of $O(1/t)$ (see Section 3.2 of \cite{bubeck2015convex}). In the proof of this result, a key step is to establish that for any $\lambda^{**} \in \argmin_{\lambda \in \mathbb{R}^E}D(\lambda)$, $\Vert \lambda[t] - \lambda^{**}\Vert$ is a nonincreasing sequence. This, coupled with the existence of a finite minimizer (Lemma \ref{lem:strong_duality}), implies the second point. The third point follows from the previous point on the convergence of $\lambda[t]$, the continuity of $F(\lambda)$ with respect to $\lambda$ (Lemma \ref{lem:smoothness}), and the fact that $f^* = F(\lambda^*)$ is primal optimal (Lemma \ref{lem:strong_duality}).
\end{proof}

\section{Simulation Results}\label{sec:pcn_simulations}
In this section, we show simulation results of the DEBT control protocol $\eqref{eq:algorithm}$ on four different instances of a PCN, illustrating how it performs dynamic routing and flow-control to ensure optimal flows in the long run. The code used to perform these experiments can be found at \url{https://github.com/sankagiri/payment-channel-networks}.

\subsection{Dynamic Routing Example}\label{sec:routing_example}
Consider a payment channel network with three nodes, $A, B$, and $C$, and three channels, $A-B$, $B-C$, and $C-A$. Each of the channels have a capacity of $100$ and they are all initially perfectly balanced. The demand is a circulation with $a_{A,B} = a_{B,C} = a_{C,A} = 10$, which can be served entirely without any on-chain rebalancing. Consider, for example, the transactions from $A$ to $B$. There are two possible paths to route such transactions: a short path directly along the channel $A-B$, and a longer path via $C$. Similarly, the other two demands also can be served over two possible paths. Observe that routing transactions along the shortest path at all times is not sustainable; it skews the channel balances. The transactions must alternate between the short path and the long path.

\begin{figure*}[!t]
\centering
\subfloat[]{\includegraphics[width=3.3in]{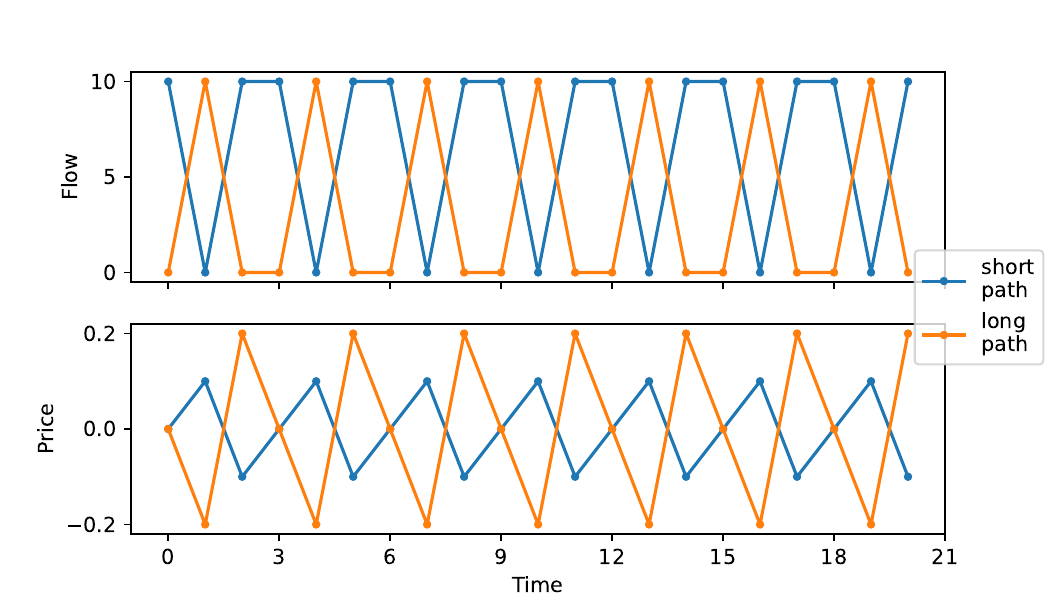}%
\label{fig:dynamic_routing_1}}
\hfil
\subfloat[]{\includegraphics[width=3.3in]{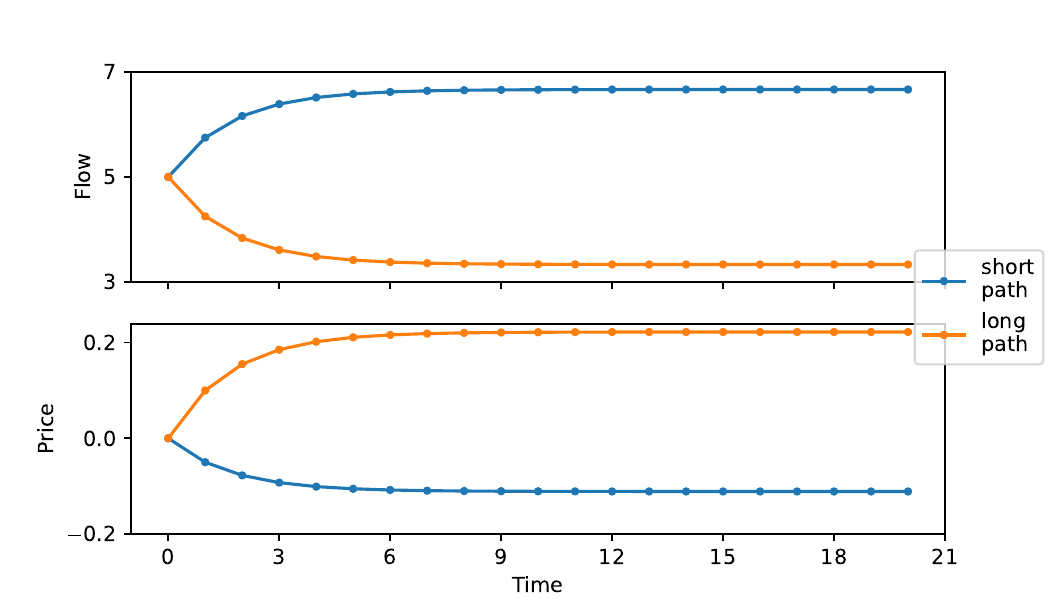}%
\label{fig:dynamic_routing_2}}
\caption{Dynamic routing ensures perennial operation in a PCN with circulation demand: an illustration of the effect of the DEBT control protocol in the PCN described in Section \ref{sec:routing_example}, (a) without the regularizer and (b) with the regularizer.}
\label{fig:dynamic_routing}
\end{figure*}

\begin{figure*}[!t]
\centering
\subfloat[]{\includegraphics[width=3.3in]{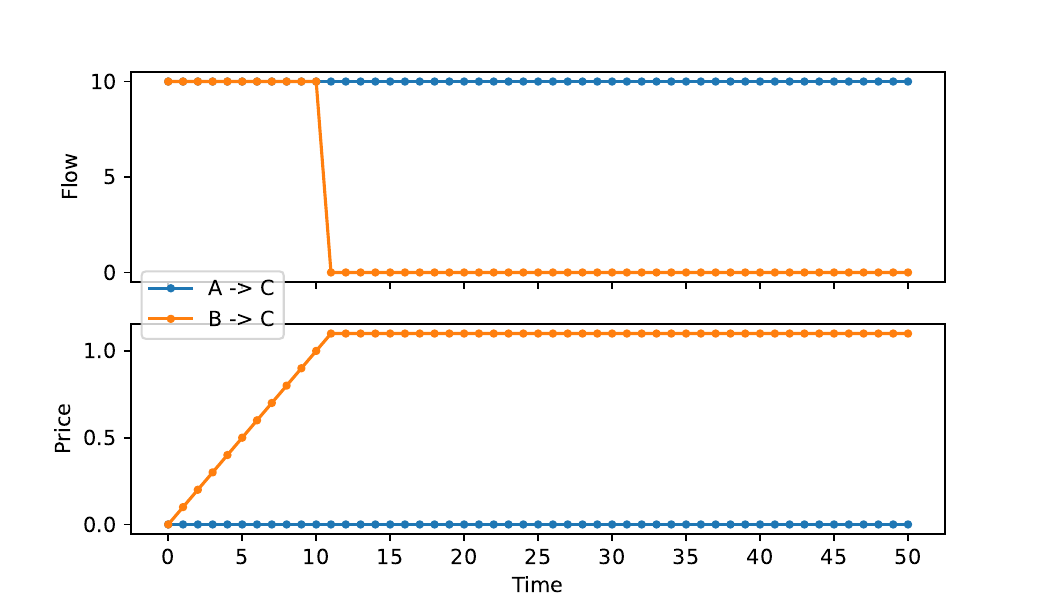}%
\label{fig:deadlock_1}}
\hfil
\subfloat[]{\includegraphics[width=3.3in]{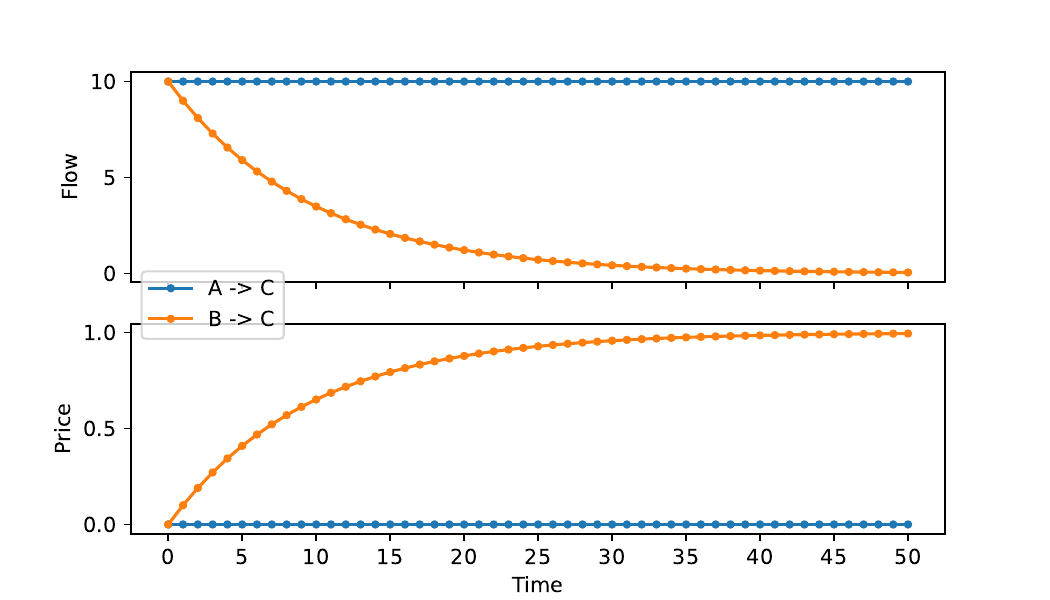}%
\label{fig:deadlock_2}}
\caption{Deadlock prevention via flow-control: an illustration of the effect of the DEBT control protocol in the PCN described in Section \ref{sec:flow_control_example}, (a) without the regularizer and (b) with the regularizer.}
\label{fig:deadlock}
\end{figure*}

Figure \ref{fig:dynamic_routing} illustrates the flows and the prices in this PCN over time under the DEBT Control protocol. 
(We only plot the flows from $A$ to $B$; the other flows are identical.) Figure \ref{fig:dynamic_routing_1} shows the case where the regularizer coefficient $\eta$ is zero, while Figure \ref{fig:dynamic_routing_2} shows the case where $\eta = 0.1$.  The stepsize $\gamma$ is set as $0.01$ in both Figures \ref{fig:dynamic_routing_1} and \ref{fig:dynamic_routing_2}. In Figure \ref{fig:dynamic_routing_1}, nodes route the entire flow along the cheaper path at all times. Ties are broken in favor of the shorter path. A flow in one direction raises prices in that direction, incentivizing nodes to take the other path in the next step. We see that the protocol follows a periodic pattern, choosing the shorter path twice in succession followed by choosing the longer path once. In Figure \ref{fig:dynamic_routing_2}, nodes split each transaction along both the paths, giving more weightage to the path with the lower price. Thus, adding the regularizer smoothens the flows as a function of the path prices. This smooth variation allows both the flows and the prices to converge. Finally, note that in both cases in Figure \ref{fig:dynamic_routing}, the long-term average of the flows satisfies the detailed balance condition.

\begin{figure*}[!t]
\centering
\subfloat[]{\includegraphics[width=3.3in]{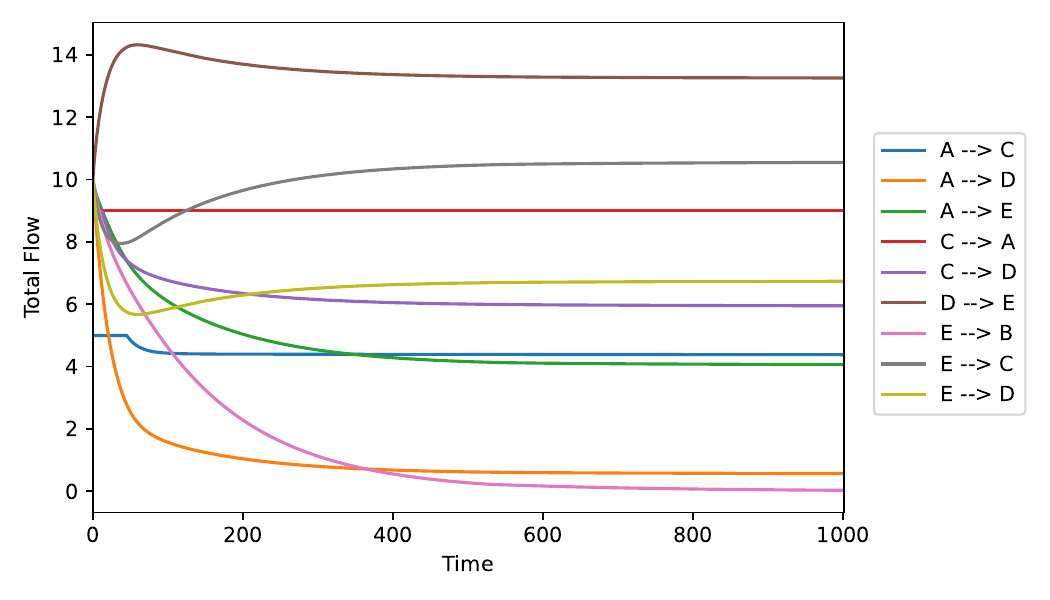}%
\label{fig:five_nodes_flow_small_steps}}
\hfil
\subfloat[]{\includegraphics[width=3.3in]{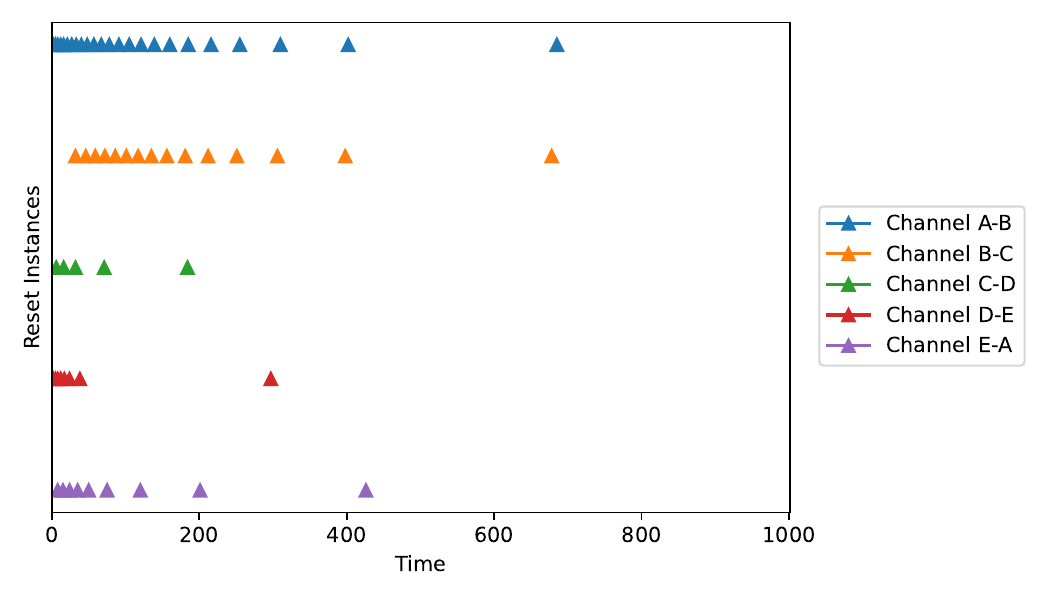}%
\label{fig:five_nodes_resets_small_steps}}
\caption{Behavior of the five node PCN with a stepsize of $\gamma = 0.01$; (a) shows flows as a function of time; (b) shows the times at which channel  resets occur.}
\label{fig:five_nodes_small}
\end{figure*}

\begin{figure*}[!t]
\centering
\subfloat[]{\includegraphics[width=3.3in]{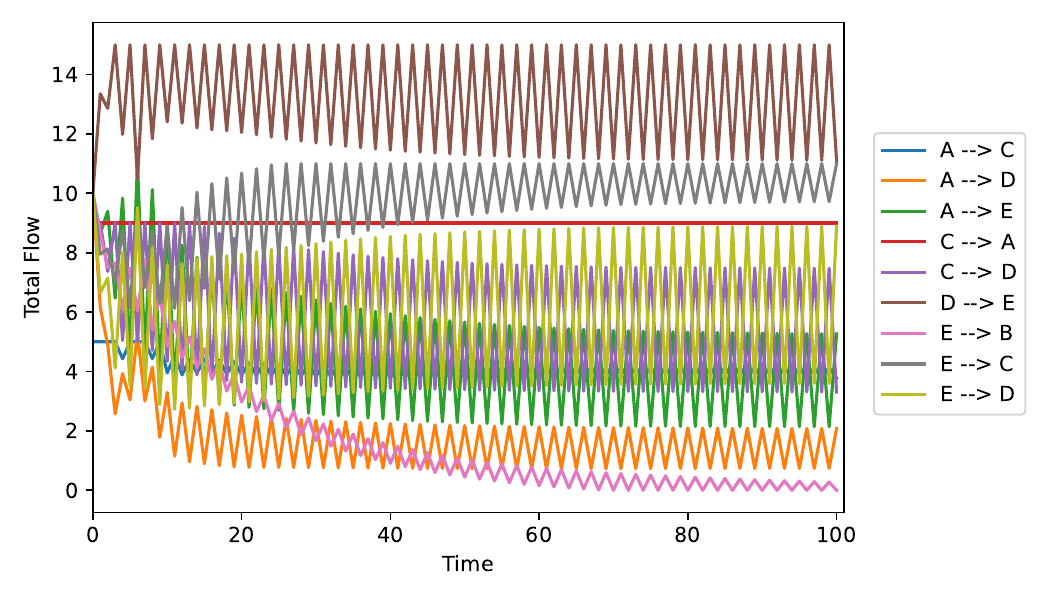}
\label{fig:five_nodes_flow_large_steps}}
\hfil
\subfloat[]{\includegraphics[width=3.3in]{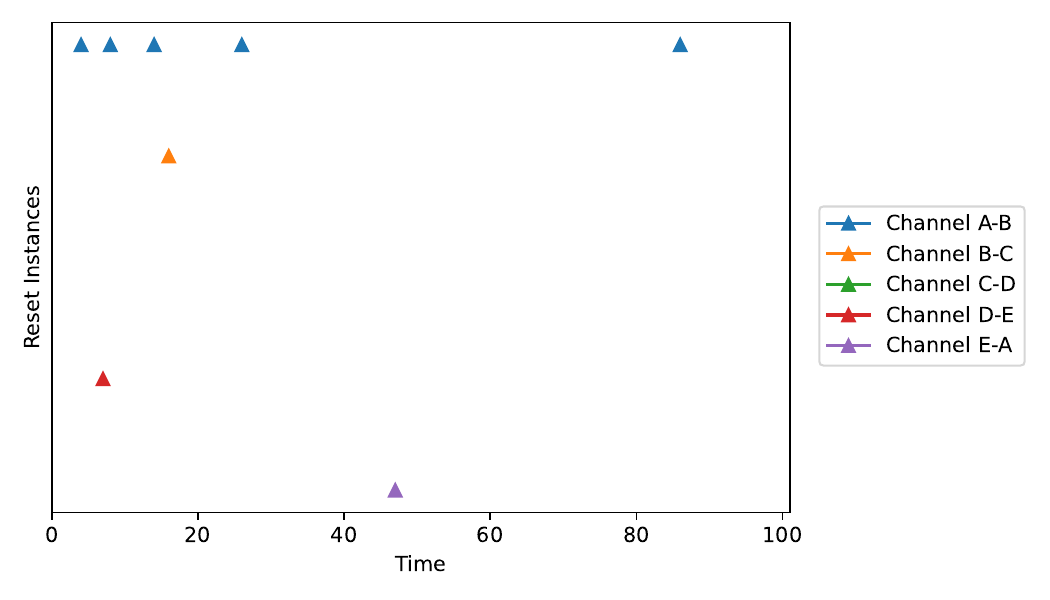}%
\label{fig:five_nodes_resets_large_steps}}
\caption{Behavior of the five node PCN with a stepsize of $\gamma = 0.1$; (a) shows flows as a function of time; (b) shows the times at which channel  resets occur.}
\label{fig:five_nodes_large}
\end{figure*}

\subsection{Flow-Control Example}\label{sec:flow_control_example}
We now demonstrate the ability of algorithm $\eqref{eq:algorithm}$ to perform flow-control by considering its behavior on a simple PCN prone to deadlocks. This example is taken from \cite{sivaraman2021effect}. The PCN has three nodes, $A, B,$ and $C$, and two channels: $A-B$ and $B-C$, with a capacity of $100$ each (initially balanced). The demands are: $a_{A,C} = a_{C, A} = a_{B,A} = a_{B,C} = 10.$ In this example, the demand from $A$ to $C$ and back form a circulation and can be sustained forever, whereas the demand from $B$ to $A$ and $C$ are DAG demands and therefore cannot be sustained. Moreover, if the network tries to serve the entire demand, eventually the balances in the two channels will get skewed towards $A$ and $C$, with no balance left at $B$. In such a state, all four flows are rendered infeasible (see \eqref{eq:feasibility} for the definition of feasibility), thus creating a deadlock.

Figure \ref{fig:deadlock} shows how algorithm \eqref{eq:algorithm} avoids this deadlock by curbing the flow from $B$ to $A$ and $C$ via channel prices. Due to symmetry, we only show the flows and path prices from $A$ to $C$ and from $B$ to $C$. The corresponding quantities from $C$ to $A$ and from $B$ to $A$ are identical to these. The stepsize $\gamma$ is chosen to be $0.01$. Once again, we show two cases, without the quadratic regularizer ($\eta = 0$) in Figure \ref{fig:deadlock_1} and with the regularizer ($\eta = 0.1$) in Figure \ref{fig:deadlock_2}. Further, we choose the utility function to be $U(\totalflow) = \totalflow$ for all the flows. Adopting this utility function means that for any path, if the price is strictly above $1$, the corresponding flow will be zero. In Figure \ref{fig:deadlock_1}, we see that with every time step, the path prices from $B$ to $A$ and $B$ to $C$ keeps increasing linearly and at some point, they exceed one. At this point, the corresponding flows turn off, which keeps the prices stable. The channel price in the $A\rightarrow B$ direction is negative, and for $B-C$ in the $B\rightarrow C$ direction is positive, which implies the prices from $A \rightarrow C$ and $C \rightarrow A$ add up to zero at all times. Thus, these flows continue unabated, as desired. The effect of the regularizer can be seen by comparing Figures \ref{fig:deadlock_1} and \ref{fig:deadlock_2}. In the absence of the regularizer, the flows exhibit a switching behavior as a function of the price. In contrast, with a regularizer, the flows vary smoothly as a function of the path price. However, the asymptotic behavior of the protocol is the same in both cases.

\subsection{Five-Node Network Example}\label{sec:five_node}
{Consider a payment channel network with five nodes ($A, B, C, D, E$) and five channels ($A$-$B$, $B$-$C$, $C$-$D$, $D$-$E$, $E$-$A$), arranged as a ring. Each of the channels has a capacity of a hundred. Because of the topology, each pair of nodes has two possible paths to transact along. The demand matrix is as follows (source along rows, destination along columns):
\begin{align*}
    \begin{array}{cccccc}
        S/D & A & B & C & D & E \\
        A & 0 & 0 & 5 & 10 & 11 \\
        B & 0 & 0 & 0 & 0 & 0 \\
        C & 9 & 0 & 0 & 9 & 0 \\
        D & 0 & 0 & 0 & 0 & 15 \\
        E & 0 & 10 & 11 & 13 & 0 \\ 
    \end{array}
\end{align*}
Note that the demands are much smaller than the channel capacities.
Each of the transacting node pairs has the same utility function $U(f) = 5f$. They also share the same regularizer coefficient $\eta = 1$.}

\begin{figure*}[!t]
\centering
\subfloat[]{\includegraphics[width=3.3in]{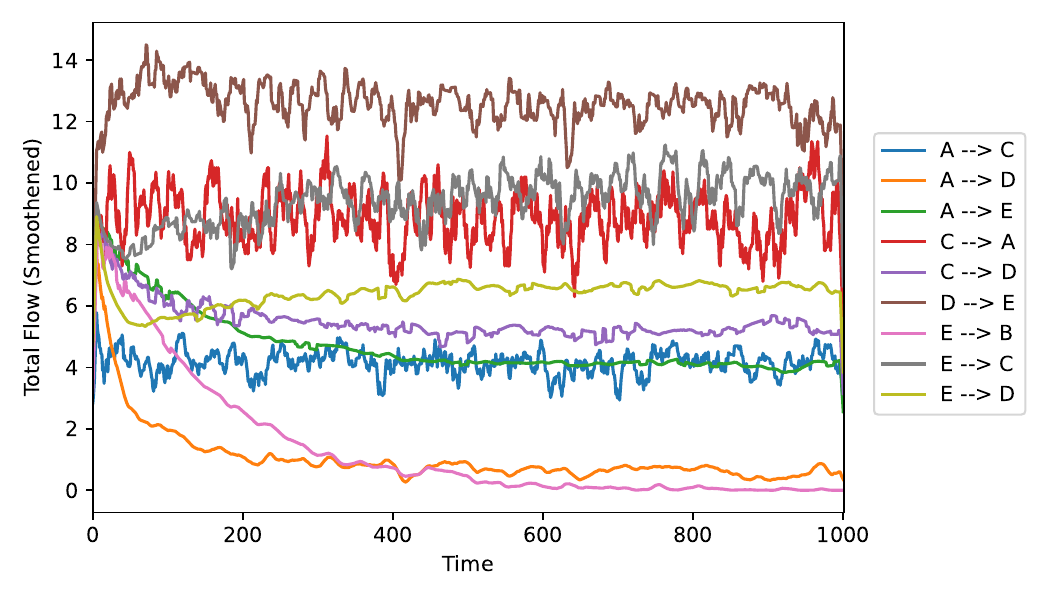}
\label{fig:five_nodes_flow_random_demand}}
\hfil
\subfloat[]{\includegraphics[width=3.3in]{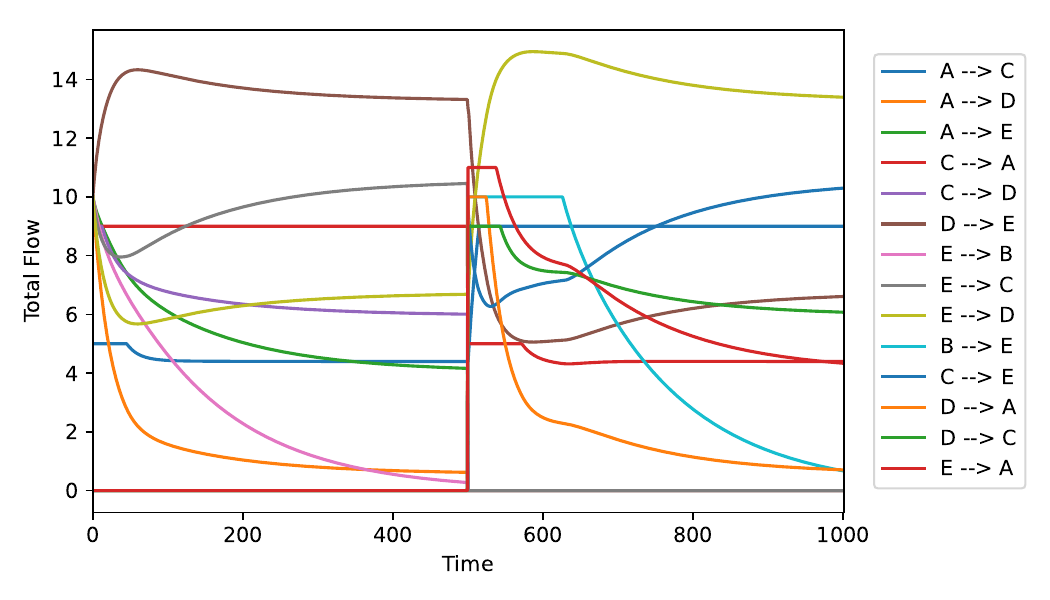}%
\label{fig:five_nodes_flows_sudden_demand}}
\caption{Flows in the five node PCN with time-varying demand; (a) when the demand at each step is independent and Poisson distributed; (b) when the demand is piece-wise constant, with a sudden change in the middle of the time horizon.}
\label{fig:five_nodes_variable_demand}
\end{figure*}

We simulate the protocol under two different settings: once with a small stepsize of $\gamma = 0.01$ (see Figure \ref{fig:five_nodes_small}) and once with a large stepsize of $\gamma = 0.1$ (see Figure \ref{fig:five_nodes_large}). In each setting, we plot the total flow between each of the transacting node pairs (Figures \ref{fig:five_nodes_flow_small_steps} and \ref{fig:five_nodes_flow_large_steps}) as well as the instances of channel resets (Figures \ref{fig:five_nodes_resets_small_steps} and \ref{fig:five_nodes_resets_large_steps}). The figures illustrate that when the step size is small, the flows converge to a stationary point. This behavior is consistent with Proposition \ref{prop:convergence}. However, the transient period is long ($\approx$ five hundred steps), which results in a fairly large number of channel resets. In contrast, when the step size is large, the flows do not converge to a single point; rather, they oscillate about a particular value. (The channel prices, not shown, also oscillate around a fixed value.) However, the protocol reaches a `steady state' much faster; note the difference in scales of the time axis in the two plots. Consequently, the number of channel resets are also smaller. In summary, varying the stepsize allows us to trade-off the smoothness of flows with the number of channel resets.

In addition, we also demonstrate the performance of the DEBT control protocol on the same network, but with variable demands in Figure \ref{fig:five_nodes_variable_demand}. First, we consider a setting with stochastic demand. The demand between any two nodes is Poisson distributed with the same mean as in the earlier simulation, and is independent at each time slot. The flows shown in Figure \ref{fig:five_nodes_flow_random_demand} are smoothened with a ten-step moving average filter. In comparison to \ref{fig:five_nodes_flow_small_steps}, the flows have similar, albeit slightly smaller (average) values. Second, we simulate a setting with a piecewise constant demand. The initial demand is the same as that of Figure \ref{fig:five_nodes_small}. Midway through the execution, the demand between every pair of nodes reverses (the source becomes the destination and vice-versa). The resulting flows are shown in Figure \ref{fig:five_nodes_flows_sudden_demand}. The plots illustrate that the protocol can quickly adapt to changes in the demand.

\subsection{Ten Node Network Example}

    The simulations in this section are based on the code provided by \cite{varma2021throughput}, allowing for a comparison of performance with the algorithm used there. We make two changes to the code. First, we change the transaction demand so it is not a pure circulation as used by \cite{varma2021throughput}. Second, we incorporate our pricing algorithm. We highlight these points below.

    The simulations are for random networks with $n=10$ nodes.  There is a bidirectional edge between any pair of nodes with probability $p=0.3.$ Between each pair of nodes, the $K$ shortest paths are considered. There are 1500 transactions between every pair of nodes. In the simulation, transaction requests are considered one at a time. The mean size of transactions from $i$ to $j$ is given by $P_{i,j}$, and the distribution of the sizes of transactions from $i$ to $j$ has the Poisson distribution with mean $P_{i,j}.$ In \cite{varma2021throughput}, the matrix $P$ was generated as a sum of $3n$ independent uniformly distributed random permutation matrices. This resulted in the demand being a circulation and the mean transaction size being three. In this work, we generate $P$ such that 
    $E[P_{i,j}] = 4$ if $i < j$ and $E[P_{i,j}] = 2$ if $i > j$. Thus, the transaction size averaged over all pairs is still 3. Observe that 2/3 of the load is a circulation and 1/3 is acyclic.  

    In \cite{varma2021throughput}, $M$ is the capacity of each link for each direction, so $2M$ is the same as the link capacity $c$ in this paper. The algorithm of  \cite{varma2021throughput} drops all demand which would make a balance get larger than $M$, so it keeps the backlog in each direction less than or equal to M. Thus, there is no queuing in the implementation; transactions are either executed immediately or dropped. In \cite{varma2021throughput} for a link $(u,v),$ if the queue sizes are $q_{(u,v)}$ and $q_{(v,u)}$ each in $[0,M]$ then an amount of transaction $s = \min\{ q_{(u,v)}, q_{(v,u)} \}$ could be locally cleared over the link.  Reducing both $q_{(u,v)}$ and $q_{(v,u)}$ by $s$ makes one of them zero.   The state $(q_{(u,v)}, q_{(v,u)} )$ in \cite{varma2021throughput} is equivalent to the channel having no queueing and total escrow amount $2M$ and the balance of escrow fund at $u$ being $x_{u,v} = M-q_{(u,v)} + q_{(v,u)},$ where $x$ is the state variable in this paper.

    Our model involves congestion control and focuses on utility maximization with utility functions $U_{i,j}$ such that $U_{i,j}(q)$ is the utility of routing quantity $q$ from $i$ to $j$ in one time slot. To focus on the amount of successful transactions and to take into account the demand $a_{i,j}$ for a transaction from $i$ to $j$ we assume that the utility functions are the following: $U_{i,j}(q) = \min\{a_{i,j},q\}\tau$ for all $i,j$, where $\tau$ is a positive constant.  If the total path price is $p$ per unit flow then the payoff for flow $i,j$ is maximized by submitting the full transaction if $p \leq \tau$ and withholding it otherwise. Thus, $\tau$ becomes a path price threshold.    

    In order to make a fair comparison with the algorithm of \cite{varma2021throughput}, we implement a version of the DEBT control protocol without rebalancing. For any transaction request, the first step is to select the path with the smallest path price (sum of link prices along the path).  If the path price exceeds the path price threshold, the transaction is not routed and prices are not updated. If the path price is less than the price threshold, then we check to see if there is sufficient capacity to route the transaction. If yes, the transaction is routed. If not, we update the prices as if the transaction was routed (i.e., as per \eqref{eq:algorithm}), but we do not update the balances; we do not count the transaction as a successful one. The theoretical model of our paper would force the transaction to be routed by forcing a complete rebalancing of the channel yielding half the escrow on each side. The version we implement is equivalent to a small rebalancing of just enough to carry the transaction. 

    In summary, the essential difference between the simulated versions of our algorithm and the algorithm of \cite{varma2021throughput} lies in the way transactions are handled on paths with insufficient balance but with price lower than the threshold. While both algorithms drop such transactions, our algorithm updates prices whereas there is no change of state for the algorithm of \cite{varma2021throughput}.

\begin{figure}[htb]
  \centering
  \begin{tabular}{c}
    \includegraphics[width=.8\linewidth]{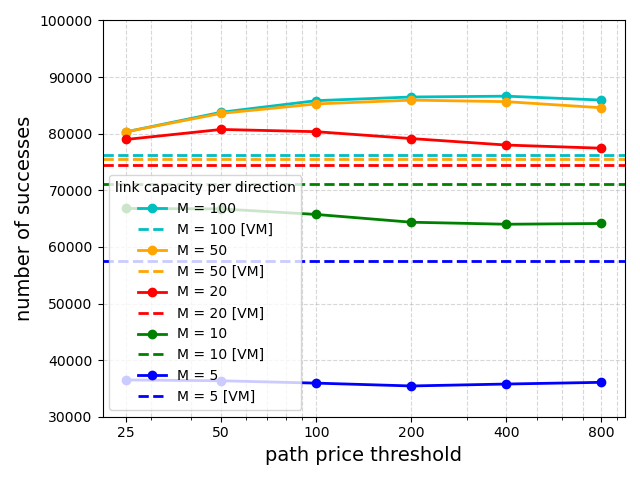}  \\
    \small (a) $K=1$ \\
    \includegraphics[width=.8\linewidth]{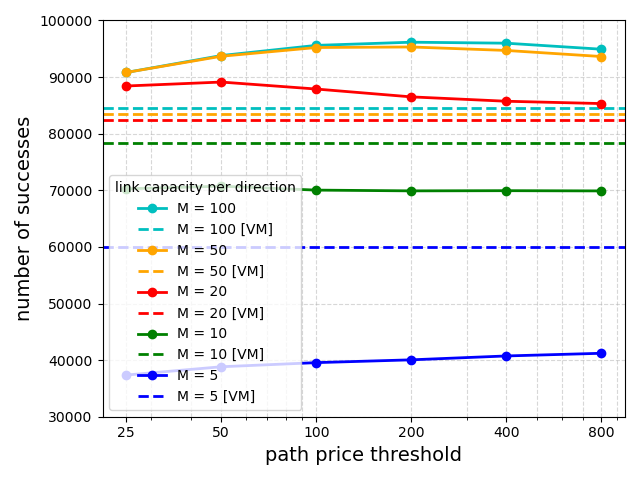}  \\
    \small (b) $K=2$ \\
  \end{tabular}
  \caption{Number of successful transactions vs path price threshold}
  \label{fig:num_success}
\end{figure}

\begin{figure}[htb]
  \centering
  \begin{tabular}{c}
    \includegraphics[width=.8\linewidth]{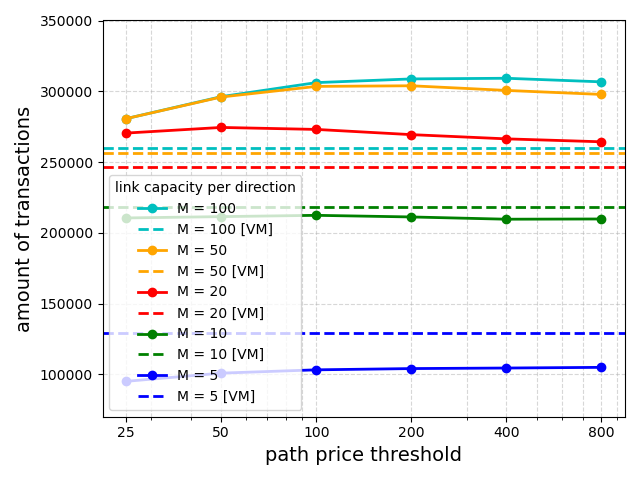}  \\
    \small (a) $K=1$ \\
    \includegraphics[width=.8\linewidth]{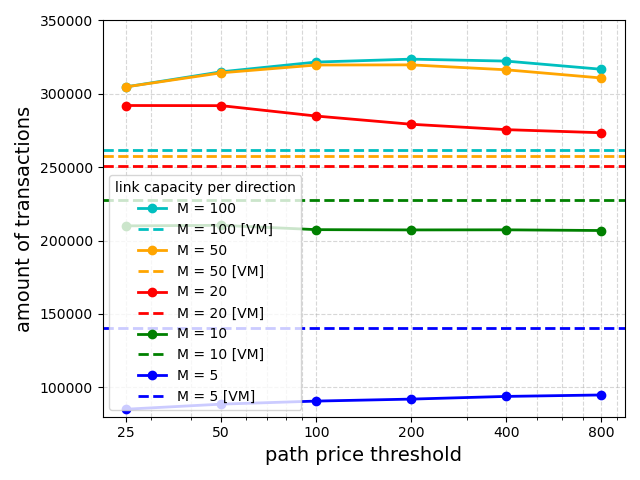}  \\
    \small (b) $K=2$ \\
  \end{tabular}
  \caption{Sum of amounts of successful transactions vs path price threshold}
  \label{fig:amount}
\end{figure}

Numerical results are shown in Figs. \ref{fig:num_success} and \ref{fig:amount}. In both Figures, the solid lines show performance of our algorithm vs. the path price threshold and the dashed lines show the performance of the algorithm of \cite{varma2021throughput}, for $M$ ranging from 5 to 100 and
$K \in \{1, 2\}.$   Fig. \ref{fig:num_success} shows the number of transactions served and Fig. \ref{fig:amount} shows the total amount of the successful transactions.  For the smallest values of $M$, $M=5$ and $M=10$ our pricing algorithm yields a significantly inferior performance.  Given that the mean transaction size is 3, these values of $M$ are considerably smaller than what is expected in applications.  For the larger values of $M$,  $M=20, 50,$ or $M=100$ our algorithm achieves significantly greater number of successes and larger amounts of transactions.  The performance of our algorithm is fairly insensitive to the value of the path price threshold.

The performance for $K=2$ is significantly better than for $K=1$ but after that we observed diminishing or negative gains with increasing $K$ further. We produced figures for $K=3$ but they are slightly lower and very close to those for $K=2$, so they are omitted.

\section{Discussion of Assumptions}\label{sec:discussion}
This paper makes several assumptions for the sake of clarity and simplicity. This section discusses the rationale behind these assumptions and the extent to which these assumptions hold in practice.

\subsection{Assumptions on the Demand}

There are two simplifying assumptions we make about the demand. First, we assume the demand at any time is relatively small compared to the channel capacities. Second, we take the demand to be constant over time. We elaborate upon both these points below.

\paragraph{Small demands} The assumption that demands are small relative to channel capacities is made precise in \eqref{eq:large_capacity_assumption}. This assumption simplifies two major aspects of our protocol. First, it largely removes congestion from consideration. In \eqref{eq:primal_problem}, there is no constraint ensuring that total flow in both directions stays below capacity--this is always met. Consequently, there is no Lagrange multiplier for congestion and no congestion pricing; only imbalance penalties apply. In contrast, protocols in \cite{sivaraman2020high, varma2021throughput, wang2024fence} include congestion fees due to explicit congestion constraints. Second, the bound \eqref{eq:large_capacity_assumption} ensures that as long as channels remain balanced, the network can always meet demand, no matter how the demand is routed. Since channels can rebalance when necessary, they never drop transactions. This allows prices and flows to adjust as per the equations in \eqref{eq:algorithm}, which makes it easier to prove the protocol's convergence guarantees. This also preserves the key property that a channel's price remains proportional to net money flow through it.

In practice, payment channel networks are used most often for micro-payments, for which on-chain transactions are prohibitively expensive; large transactions typically take place directly on the blockchain. For example, according to \cite{river2023lightning}, the average channel capacity is roughly $0.1$ BTC ($5,000$ BTC distributed over $50,000$ channels), while the average transaction amount is less than $0.0004$ BTC ($44.7k$ satoshis). Thus, the small demand assumption is not too unrealistic. Additionally, the occasional large transaction can be treated as a sequence of smaller transactions by breaking it into packets and executing each packet serially (as done by \cite{sivaraman2020high}).
Lastly, a good path discovery process that favors large capacity channels over small capacity ones can help ensure that the bound in \eqref{eq:large_capacity_assumption} holds.

\paragraph{Constant demands} 
In this work, we assume that any transacting pair of nodes have a steady transaction demand between them (see Section \ref{sec:transaction_requests}). Making this assumption is necessary to obtain the kind of guarantees that we have presented in this paper. Unless the demand is steady, it is unreasonable to expect that the flows converge to a steady value. Weaker assumptions on the demand lead to weaker guarantees. For example, with the more general setting of stochastic, but i.i.d. demand between any two nodes, \cite{varma2021throughput} shows that the channel queue lengths are bounded in expectation. If the demand can be arbitrary, then it is very hard to get any meaningful performance guarantees; \cite{wang2024fence} shows that even for a single bidirectional channel, the competitive ratio is infinite. Indeed, because a PCN is a decentralized system and decisions must be made based on local information alone, it is difficult for the network to find the optimal detailed balance flow at every time step with a time-varying demand.  With a steady demand, the network can discover the optimal flows in a reasonably short time, as our work shows.

We view the constant demand assumption as an approximation for a more general demand process that could be piece-wise constant, stochastic, or both (see simulations in Figure \ref{fig:five_nodes_variable_demand}).
We believe it should be possible to merge ideas from our work and \cite{varma2021throughput} to provide guarantees in a setting with random demands with arbitrary means. We leave this for future work. In addition, our work suggests that a reasonable method of handling stochastic demands is to queue the transaction requests \textit{at the source node} itself. This queuing action should be viewed in conjunction with flow-control. Indeed, a temporarily high unidirectional demand would raise prices for the sender, incentivizing the sender to stop sending the transactions. If the sender queues the transactions, they can send them later when prices drop. This form of queuing does not require any overhaul of the basic PCN infrastructure and is therefore simpler to implement than per-channel queues as suggested by \cite{sivaraman2020high} and \cite{varma2021throughput}.

\paragraph{Bounded derivative of utility functions}
We comment briefly on our assumption that the utility functions satisfy $U'_{i,j}(0)<\infty$ and the choice of the regularizer $H(f)$ being quadratic.  Other concave functions could be used for $H$ but it is important that $H$ has bounded gradient over the set of possible flow vectors.  These properties on the $U_{i,j}$'s and $H$ are used for the third bulleted condition in the proof of Lemma \ref{lem:strong_duality} for the application of Proposition 3.4.2 of \cite{bertsekas1999nonlinear}.   In particular, the proof does not go  through for
logarithmic utility functions used in \cite{sivaraman2020high} or logarithmic/entropic regularizers.   The price convergence in Proposition \ref{prop:convergence} doesn't hold for such utility functions or regularizers because no matter how large the price on a path is, the flow that maximizes the regularized utility will assign a nonzero flow to the path.  So the path price would need to converge to $+\infty$ to drive the flow to zero, which is necessary in some situations to prevent deadlock, as in Section 
\ref{sec:flow_control_example} based on \cite{sivaraman2021effect}.

\subsection{The Incentive of Channels}
The actions of the channels as prescribed by the DEBT control protocol can be summarized as follows. Channels adjust their prices in proportion to the net flow through them. They rebalance themselves whenever necessary and execute any transaction request that has been made of them. We discuss both these aspects below.

\paragraph{On Prices}
In this work, the exclusive role of channel prices is to ensure that the flows through each channel remains balanced. In practice, it would be important to include other components in a channel's price/fee as well: a congestion price  and an incentive price. The congestion price, as suggested by \cite{varma2021throughput}, would depend on the total flow of transactions through the channel, and would incentivize nodes to balance the load over different paths. The incentive price, which is commonly used in practice \cite{river2023lightning}, is necessary to provide channels with an incentive to serve as an intermediary for different channels. In practice, we expect both these components to be smaller than the imbalance price. Consequently, we expect the behavior of our protocol to be similar to our theoretical results even with these additional prices.

We have assumed all channels use the same step size $\gamma$ in updating their prices and feel this is not difficult to enforce in practice.   However, if different channels used sufficiently small but different step sizes the updates would be quasi-Newton gradient updates with constant diagonal scaling matrix and would still achieve global convergence.

A key aspect of our protocol is that channel fees are allowed to be negative. Although the original Lightning network whitepaper \cite{poon2016bitcoin} suggests that negative channel prices may be a good solution to promote rebalancing, the idea of negative prices in not very popular in the literature. To our knowledge, the only prior work with this feature is \cite{varma2021throughput}. Indeed, in papers such as \cite{van2021merchant} and \cite{wang2024fence}, the price function is explicitly modified such that the channel price is never negative. The results of our paper show the benefits of negative prices. For one, in steady state, equal flows in both directions ensure that a channel doesn't loose any money (the other price components mentioned above ensure that the channel will only gain money). More importantly, negative prices are important to ensure that the protocol selectively stifles acyclic flows while allowing circulations to flow. Indeed, in the example of Section \ref{sec:flow_control_example}, the flows between nodes $A$ and $C$ are left on only because the large positive price over one channel is canceled by the corresponding negative price over the other channel, leading to a net zero price.

Lastly, observe that in the DEBT control protocol, the price charged by a channel does not depend on its capacity. This is a natural consequence of the price being the Lagrange multiplier for the net-zero flow constraint, which also does not depend on the channel capacity. In contrast, in many other works, the imbalance price is normalized by the channel capacity \cite{ren2018optimal, lin2020funds, wang2024fence}; this is shown to work well in practice. The rationale for such a price structure is explained well in \cite{wang2024fence}, where this fee is derived with the aim of always maintaining some balance (liquidity) at each end of every channel. This is a reasonable aim if a channel is to never rebalance itself; the experiments of the aforementioned papers are conducted in such a regime. In this work, however, we allow the channels to rebalance themselves a few times in order to settle on a detailed balance flow. This is because our focus is on the long-term steady state performance of the protocol. This difference in perspective also shows up in how the price depends on the channel imbalance. \cite{lin2020funds} and \cite{wang2024fence} advocate for strictly convex prices whereas this work and \cite{varma2021throughput} propose linear prices.

\paragraph{On Rebalancing} 
Recall that the DEBT control protocol ensures that the flows in the network converge to a detailed balance flow, which can be sustained perpetually without any rebalancing. However, during the transient phase (before convergence), channels may have to perform on-chain rebalancing a few times. Since rebalancing is an expensive operation, it is worthwhile discussing methods by which channels can reduce the extent of rebalancing. One option for the channels to reduce the extent of rebalancing is to increase their capacity; however, this comes at the cost of locking in more capital. Each channel can decide for itself the optimum amount of capital to lock in. Another option, which we discuss in Section \ref{sec:five_node}, is for channels to increase the rate $\gamma$ at which they adjust prices. 

Ultimately, whether or not it is beneficial for a channel to rebalance depends on the time-horizon under consideration. Our protocol is based on the assumption that the demand remains steady for a long period of time. If this is indeed the case, it would be worthwhile for a channel to rebalance itself as it can make up this cost through the incentive fees gained from the flow of transactions through it in steady state. If a channel chooses not to rebalance itself, however, there is a risk of being trapped in a deadlock, which is suboptimal for not only the nodes but also the channel.

\section{Conclusion}
This work presents DEBT control: a protocol for payment channel networks that uses source routing and flow control based on channel prices. The protocol is derived by posing a network utility maximization problem and analyzing its dual minimization. It is shown that under steady demands, the protocol guides the network to an optimal, sustainable point. Simulations show its robustness to demand variations. The work demonstrates that simple protocols with strong theoretical guarantees are possible for PCNs. 

\printbibliography
\end{document}